\documentclass[english,dvipsnames]{article}
\usepackage[T1]{fontenc}
\usepackage[latin9]{inputenc}
\usepackage{geometry}
\usepackage{xcolor}
\usepackage{babel}
\usepackage{amsmath}
\usepackage{amsthm}
\usepackage{amssymb}
\usepackage{graphicx}
\usepackage[authoryear]{natbib}
\PassOptionsToPackage{normalem}{ulem}
\usepackage{ulem}
\usepackage[unicode=true]
 {hyperref}

\makeatletter

\providecommand{\tabularnewline}{\\}

\theoremstyle{plain}
\newtheorem{thm}{\protect\theoremname}[section]
\theoremstyle{plain}
\newtheorem{assumption}[thm]{\protect\assumptionname}
\theoremstyle{plain}
\newtheorem{lem}[thm]{\protect\lemmaname}

\usepackage{hyperref}
\usepackage{authblk}

\usepackage{xcolor}

\definecolor{mygray}{rgb}{0.6,0.6,0.6}

\newcommand{\mred}[1]{\textcolor{red}{\ensuremath{#1}}}
\newcommand{\mblue}[1]{\textcolor{blue}{\ensuremath{#1}}}
\newcommand{\mpur}[1]{\textcolor{purple}{\ensuremath{#1}}}
\newcommand{\tpur}[1]{\textcolor{purple}{#1}}

\makeatletter
\newcommand*\bigcdot{\mathpalette\bigcdot@{.8}}
\newcommand*\bigcdot@[2]{\mathbin{\vcenter{\hbox{\scalebox{#2}{$\m@th#1\bullet$}}}}}
\makeatother

\author[1]{Aaron Fisher}
\affil[1]{Foundation Medicine Inc.}

\usepackage{microtype}
\usepackage{graphicx}
\usepackage{subfigure}
\usepackage{booktabs} 

\usepackage{hyperref}

\usepackage{amsmath}
\usepackage{amssymb}
\usepackage{mathtools}
\usepackage{amsthm}

\usepackage[capitalize,noabbrev]{cleveref}

\usepackage[textsize=tiny]{todonotes}

\makeatother

\providecommand{\assumptionname}{Assumption}
\providecommand{\lemmaname}{Lemma}
\providecommand{\theoremname}{Theorem}

\begin{document}
\twocolumn
\title{The Connection Between R-Learning and Inverse-Variance Weighting for
Estimation of Heterogeneous Treatment Effects}
\maketitle
\begin{abstract}
Many methods for estimating conditional average treatment effects
(CATEs) can be expressed as weighted pseudo-outcome regressions (PORs).
Previous comparisons of POR techniques have paid careful attention
to the choice of pseudo-outcome transformation. However, we argue
that the dominant driver of performance is actually the choice of
\emph{weights.} For example, we point out that R-Learning implicitly
performs a POR with \emph{inverse-variance} \emph{weights} (IVWs).
In the CATE setting, IVWs mitigate the instability associated with
inverse-propensity weights, and lead to convenient simplifications
of bias terms. We demonstrate the superior performance of IVWs in
simulations, and derive convergence rates for IVWs that are, to our
knowledge, the fastest yet shown without assuming knowledge of the
covariate distribution.
\end{abstract}

\section{Introduction}

Estimates of conditional average treatment effects (CATEs) allow for
treatment decisions to be tailored to the individual. Formally, let
$A\in\{0,1\}$ be a binary treatment, let $X\in\mathcal{X}$ be a
vector of confounders and treatment effect modifiers, let $Y_{(a)}$
be the potential outcome under treatment $a$, and let $Y=AY_{(1)}+(1-A)Y_{(0)}$
be the observed outcome. The CATE is defined as $\tau(X):=\mathbb{E}\left(Y_{(1)}-Y_{(0)}|X\right)$.
Under conventional assumptions of exchangeability and positivity,\footnote{That is, $Y_{(1)},Y_{(0)}\perp A|X$ and $\text{Pr}(A=1|X)\in(c,1-c)$
for some $c\in(0,1)$.} the CATE can be identified as $\tau(x)=\mathbb{E}\left(Y|X,A=1\right)-\mathbb{E}\left(Y|X,A=0\right)$.

CATE estimation has a rich history going back several decades (see,
e.g., \citealp{Robins1995-aw,Hill2011-zd,Zhao2012-qx,Imai2013-bz,Athey2016-vj,Hahn2017-hi}).
We focus here on two general approaches: pseudo-outcome regression
(POR) and R-learning. Both approaches easily accommodate flexible
machine learning tools, and can attain double robustness (DR) properties
similar to those established in the average treatment effect (ATE)
literature (\citealp{Kennedy2022-da,Nie2020-ih}; see also \citealt{Scharfstein1999-qw,Robins2000-xh,Bang2005-kd,Chernozhukov2022-iq,Kennedy2022-jv})

POR aims to derive a noisy but unbiased approximation of $Y_{(1)}-Y_{(0)}$,
and to fit a regression to predict this approximation using $X$ (\citealp{Rubin2005-pe,Van_der_Laan2006-os,Tian2014-fz,Chen2017-uv,Foster2019-nb,Kunzel2019-ni,Semenova2020-bw,Curth2021-qi};
see also \citealt{Buckley1979-gg,Fan1994-vz,Rubin2007-ka,Diaz2018-ax}).
The approximation of $Y_{(1)}-Y_{(0)}$ is referred to as a ``unbiasing
transformation'' or ``pseudo-outcome'' because it serves as an
observed stand-in for the latent outcome of interest $Y_{(1)}-Y_{(0)}$.
For example, if the propensity scores $\text{Pr}(A=1|X)$ are known,
then an appropriate pseudo-outcome can be derived using inverse propensity
weights: $f_{\text{IPW}}(A,Y):=AY/\text{Pr}(A=1)-(1-A)Y/\text{Pr}(A=0)$.
Since $\mathbb{E}\left(f_{\text{IPW}}(A,Y)|X\right)=\tau(X)$, regressing
the pseudo-outcomes $f_{\text{IPW}}(A,Y)$ against $X$ produces a
sensible estimate of $\tau$ \citep{Powers2018-ya}. This regression
can be done with any off-the-shelf machine learning algorithm. For
this reason, POR methods are sometimes referred to as \textquotedblleft meta-algorithms\textquotedblright{}
(\citealp{Kennedy2022-da}).

R-learning estimates the CATE using a moment condition derived by
\citeauthor{Robinson1988-vy} (\citeyear{Robinson1988-vy}; see Section
5.2 of \citealp{Robins2008-lt}; \citealp{Semenova2017-ro,Nie2020-ih,Zhao2022-fn};
\citealp{Kennedy2022-da,Kennedy2022-cn}). While R-Learning is sometimes
described as separate from POR, it can also be expressed as a \emph{weighted
}POR (see Section \ref{sec:Motivation-intuition}, below, and the
\uline{\href{https://econml.azurewebsites.net/spec/estimation/dml.html}{NonParamDML}}
method in the EconML package from \citealt{Syrgkanis2021-pd}). 

This parallel between R-learning and weighted POR invites the question
of whether or not weights should be used in POR more broadly, and,
if so, what choice of weights is optimal? In other words, even after
confounding bias has been accounted for through a pseudo-outcome transformation
(e.g., $f_{\text{IPW}}$), should \emph{additional} weights be used
to prioritize the fit of $\tau$ of different subregions of $\mathcal{X}$?
We aim to shed light on this question through a combination of simulation
\& theory.

\subsection*{Contribution Summary}

The main intuition of this manuscript is that pseudo-outcomes based
on inverse-propensity weights are effective at removing confounding,
but can be unstable in the face of propensity scores close to zero
or one. Inverse-\emph{variance} weights restabilize the POR without
reintroducing confounding, since the CATE estimand is conditional
on $X$, and $Y$ is unconfounded within strata of $X$. This form
of reweighting is done implicitly by the R-Learner. 

Section \ref{sec:Motivation-intuition} discusses the above intuition
in more detail. Section \ref{sec:Simulations} shows that the intuition
bears out in simulations. Section \ref{sec:theory} demonstrates how
the framework of weighted POR can be used to study bias terms for
CATE estimates, and to derive fast convergence rates. We close with
a discussion.

\subsection{Stabilizing weights in CATE estimation\label{sec:Motivation-intuition}}

In this section we outline connections between R-Learning and inverse-variance
weighting (IVW). Let $Z:=(Y,X,A)$, and let
\begin{align*}
\mu_{a}(X) & =\mathbb{E}\left(Y|X,A=a\right),\\
\eta(X) & =\mathbb{E}\left(Y|X\right),\\
\pi(X) & =\text{Pr}\left(A=1|X\right),\\
\kappa(X) & =\text{Pr}\left(A=0|X\right),\text{ and}\\
\nu(X) & =Var(A|X).
\end{align*}
Let $\theta=\{\mu_{1},\mu_{0},\eta,\pi,\kappa,\nu\}$ denote the full
vector of nuisance functions, and let $\hat{\theta}=\{\hat{\mu}_{1},\hat{\mu}_{0},\hat{\eta},\hat{\pi},\hat{\kappa},\hat{\nu}\}$
be a set of corresponding nuisance estimates. We use $\mu$ and $\hat{\mu}$
as shorthand for $\{\mu_{0},\mu_{1}\}$ and $\{\hat{\mu}_{0},\hat{\mu}_{1}\}$
respectively. One of the reasons we include the redundant representations
$\pi(x)$ and $\kappa(x)=1-\pi(x)$ is to simplify certain formulas
and bias results later on. The notation ``kappa'' is meant to be
reminiscent of the term ``\textbf{\emph{c}}ontrol.'' 

\subsubsection{Weights used in R-Learning}

Given a pair of pre-estimated nuisance functions $\hat{\eta}$ and
$\hat{\pi}$, the R-Learning estimate of the CATE ($\tau$) is typically
written as
\begin{align}
\arg\min_{\hat{\tau}} & \sum_{i=1}^{n}\left[\left\{ Y_{i}-\hat{\eta}(X_{i})\right\} -\left\{ A_{i}-\hat{\pi}(X_{i})\right\} \hat{\tau}(X_{i})\right]^{2}.\label{eq:RL-def-min}
\end{align}
The procedure is motivated by the fact that the term in square brackets
has mean zero when $\hat{\eta}=\eta$, $\hat{\pi}=\pi$ and $\hat{\tau}=\tau$
\citep{Robinson1988-vy}. The nuisance estimates $\hat{\eta}$ and
$\hat{\pi}$, are typically obtained via \emph{cross-fitting }(CF):
splitting the sample into two partitions, using one to estimate $\hat{\eta}$
and $\hat{\pi}$, and using the other to create the summands in Eq
(\ref{eq:RL-def-min}) (\citealp{Nie2020-ih,Kennedy2020-ie,Kennedy2022-jv,Chernozhukov2022-jx,Chernozhukov2022-iq};
see also related work from, e.g., \citealt{Bickel1982-ht,Schick1986-qk,Bickel1988-mt},
as well as \citealt{Athey2016-vj}). In general, we assume in this
section that $\hat{\theta}$ is pre-estimated from an independent
dataset or sample partition.

A known but often overlooked fact is that the minimization in Eq (\ref{eq:RL-def-min})
can equivalently be solved by fitting a \emph{weighted regression}
using $X$ to predict 
\begin{align}
f_{\text{U},\hat{\theta}}(Z):= & \frac{Y-\hat{\eta}(X)}{A-\hat{\pi}(X)}\label{eq:U-pseudo-outcome}
\end{align}
with weights $\left\{ A-\hat{\pi}(X)\right\} ^{2}$ and the squared
error loss function. While this connection is known in the literature
as a computational trick for implementing R-Learning (see, e.g., Eq
(8) of \citealp{Zhao2022-fn}; and the \uline{\href{https://econml.azurewebsites.net/spec/estimation/dml.html}{NonParamDML}}
method in the EconML package, \citealt{Syrgkanis2021-pd}), there
appears to be little discussion of how the regression framing can
serve to motivate R-Learning in the first place. 

One such motivation comes from ``U-Learning,'' a method that fits
an \emph{unweighted} regression to predict $f_{\text{U},\hat{\theta}}(Z)$
from $X$ (see the Appendix of \citealp{Kunzel2019-ni}). The rationale
for U-Learning is that, if $\hat{\pi}=\pi$ and $\hat{\eta}=\eta$,
then $f_{U,\hat{\theta}}$ is a pseudo-outcome in the sense that $\mathbb{E}\left[f_{\text{U},\hat{\theta}}(Z)|X\right]=\tau(X)$
(\citealp{Robinson1988-vy,Kunzel2019-ni,Nie2020-ih}). \footnote{This follows from the ``Robinson Decomposition.''}
This rationale immediately applies to R-Learning as well.

Moreover, we can motivate the R-Learner's weights by appealing to
the intuition of inverse-variance weighted least squares. We show
in Appendix \ref{sec:Variance-of-POs} that, if $\hat{\theta}=\theta$,
the treatment effect is null (i.e., $A\perp Y|X$), and the outcome
$Y$ is homoskedastic (i.e., $Var(Y|X)=\sigma^{2}$ is constant),
then the pseudo-outcome $f_{U,\hat{\theta}}$ used in R-Learning has
conditional variance
\begin{align}
Var\left(\frac{Y-\eta(X)}{A-\pi(X)}|X\right) & \propto\mathbb{E}\left[\left(A-\pi(X)\right)^{-2}|X\right].\label{eq:varx}
\end{align}
In this way, the weights $\left\{ A-\hat{\pi}(X)\right\} ^{2}$ used
by R-Learning are approximate IVWs, and we would expect them to stabilize
the regression. 

Indeed, \citet{Nie2020-ih} remark that U-Learning suffers from instability
due to the denominator in $f_{\text{U},\hat{\theta}}(Z)$. They find
that R-Learning generally outperforms the U-Learner in simulations.
Since the R-Learner is equivalent to a weighted U-Learner, this finding
effectively means that the $\left\{ A-\hat{\pi}(X)\right\} ^{2}$
weights used in R-Learning counteract the instabilities of U-Learning.
To our knowledge, the implicit connections between R-Learning, U-Learning
and IVW have not been discussed in the literature.

Figure \ref{fig:illustration} shows a simple simulated illustration
of how the R-Learner's weights provide stabilization. Here, $X\sim U(0.05,0.95)$,
$\pi(X)=X$, and $Y\sim N(0,1)$ regardless of the value of $(A,X)$.
This implies that $\tau(x)=0$ for all $x$, and that the propensity
score is most extreme when $x$ is close to 0 or 1. For simplicity
of illustration, we briefly assume perfect knowledge of the nuisance
functions, and use this knowledge to define pseudo-outcomes according
to Eq (\ref{eq:U-pseudo-outcome}). (We remove this assumption in
our theoretical analysis and main simulation study.) Given these pseudo-outcomes,
we apply both U-Learning and R-Learning using spline-based, (weighted)
POR. Figure \ref{fig:illustration} shows the results. Here, we can
see that values of $x$ close to 0 or 1 produce extreme propensity
scores, which lead to instability in the pseudo-outcomes. While this
hinders the U-Learner's performance, the R-Learner is able to provide
a more stable result and a lower rMSE by down-weighting observations
with extreme propensity scores.

\begin{figure*}

\begin{centering}
\includegraphics[width=0.74\paperwidth]{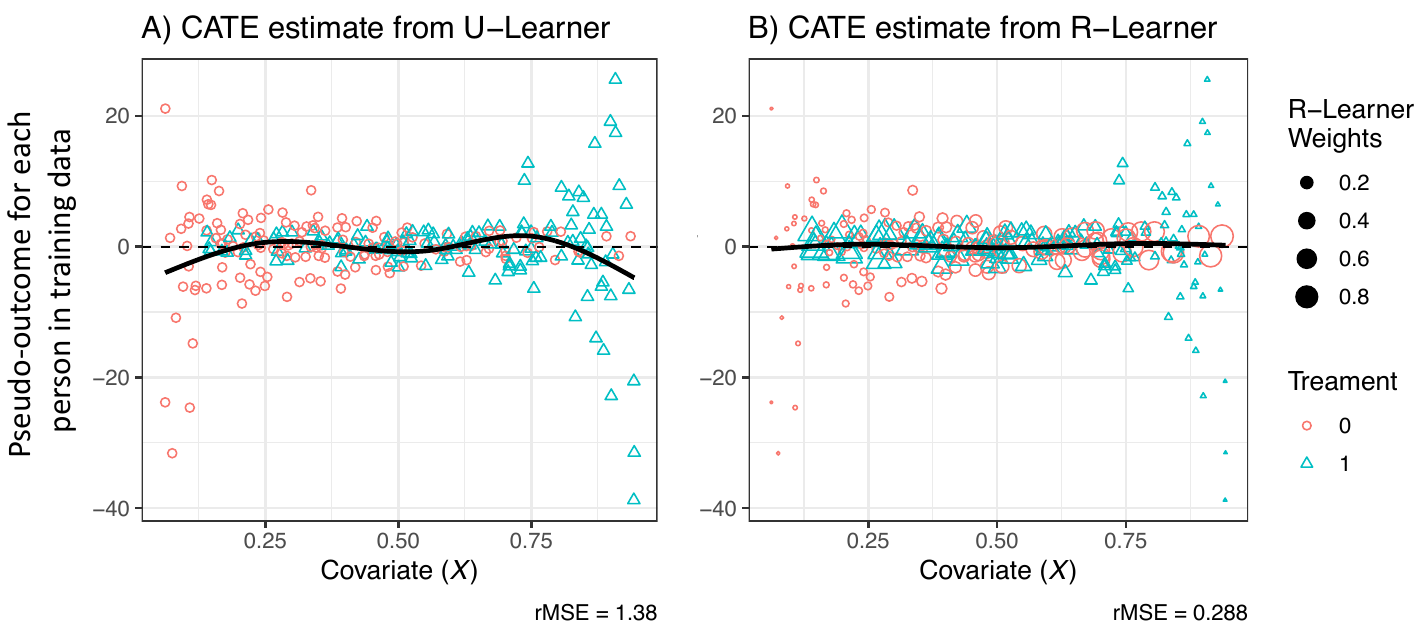}
\par\end{centering}
\caption{\label{fig:illustration}Example of how weights stabilize pseudo-outcome
regression, using a single simulated sample. Here, the true conditional
average treatment effect is zero for all patients. The estimates from
U-Learning \& R-Learning are shown as black lines. By down-weighting
the observations with high variance, i.e., those with extreme propensity
scores, R-Learning is able to achieve a lower rMSE.}

\end{figure*}

\subsubsection{Alternative motivation for R-Learner's weights}

As an alternative to Eq (\ref{eq:varx}), a similar motivation for
the R-Learner's weights can be derived by noting that $\left\{ A-\hat{\pi}(X)\right\} ^{2}$
is roughly proportional the inverse variance of $f_{\text{U,\ensuremath{\hat{\theta}}}}(Z)$
conditional on \emph{conditional on $\hat{\theta}$, $X$} \emph{and}
$A$. More specifically, if $Var\left(Y|A,X\right)=\sigma^{2}$ is
constant, then
\[
Var\left(\frac{Y-\hat{\eta}(X)}{A-\hat{\pi}(X)}|A,X,\hat{\theta}\right)\propto\left\{ A-\hat{\pi}(X)\right\} ^{-2}.
\]
Thus, if we were to expand R-Learning to predict $f_{\text{U,\ensuremath{\hat{\theta}}}}$
as a function of \emph{both $X$ and} $A$, and if $Var(Y|A,X)$ were
constant, then $\left\{ A-\hat{\pi}(X)\right\} ^{2}$ would form appropriate
inverse variance weights, producing the regression problem
\begin{equation}
\arg\min_{\hat{g}}\sum_{i=1}^{n}\left\{ A_{i}-\hat{\pi}(X_{i})\right\} ^{2}\left\{ \frac{Y-\hat{\eta}(X)}{A-\hat{\pi}(X)}-\hat{g}(A_{i},X_{i})\right\} ^{2}.\label{eq:empirical-opt}
\end{equation}
The change to include $A$ as a covariate is balanced by the fact
that, if $\hat{\theta}=\theta$, then the population minimizer for
Eq (\ref{eq:empirical-opt}), $\mathbb{E}\left[\frac{Y-\eta(X)}{A-\pi(X)}|A,X\right]$,
does not actually depend on $A$. More specifically, the Robinson
Decomposition implies that $\mathbb{E}\left[\frac{Y-\eta(X)}{A-\pi(X)}|A,X\right]=\tau(X)$.
Reflecting this fact, if we additionally require the solution to Eq
(\ref{eq:empirical-opt}) to not depend on $A$, then we recover R-Learning
exactly. 

\subsubsection{Weights in \textquotedblleft oracle\textquotedblright{} R-Learning}

A similar connection to stabilizing weights can be seen in the ``oracle''
version of R-Learning studied by \citeauthor{Kennedy2022-da} (\citeyear{Kennedy2022-da};
see their Section 7.6.1). This hypothetical oracle model fits a weighted
POR to predict the latent function 
\begin{align*}
f_{\text{OR},\theta}(Z) & :=\frac{\left\{ A-\pi(X)\right\} \left\{ Y-\eta(X)\right\} }{\pi(X)\left\{ 1-\pi(X)\right\} }\\
 & \approx\frac{\left\{ A-\hat{\pi}(X)\right\} \left\{ Y-\hat{\eta}(X)\right\} }{\left\{ A-\hat{\pi}(X)\right\} ^{2}}\\
 & =f_{\text{U},\hat{\theta}}(A,X,Y),
\end{align*}
with weights $\nu(X)=Var(A|X)$. Above, the approximation simply reflects
the fact that if $\hat{\pi}=\pi$ then the conditional expectation
of the denominators are identical. Again, if the treatment effect
is null ($A\perp Y|X$) and the conditional variance of $Y$ is constant
(i.e., $Var(Y|X)=\sigma^{2}$), then 
\[
Var(f_{\text{OR},\theta}(A,X,Y)|X)\propto\nu(X)^{-1}
\]
(see Appendix \ref{sec:Variance-of-POs}). Thus, in the null setting,
the oracle R-Learner is an inverse-variance weighted POR.

\subsubsection{Weights for the DR-Learner}

Another pseudo-outcome transformations that can suffer from instability
is the ``DR-Learner'' (\citealp{Kennedy2022-da}). This method fits
a regression using $X$ to predict $f_{\text{DR},\hat{\theta}}(Z)=f_{1,\hat{\theta}}(Z)-f_{0,\hat{\theta}}(Z)$,
where
\begin{align}
 & f_{a,\hat{\theta}}(Z)\nonumber \\
 & \hspace{1em}=\hat{\mu}_{a}(X)+\frac{1(A=a)}{a\hat{\pi}(X)+(1-a)\hat{\kappa}(X)}\left(Y-\hat{\mu}_{a}(X)\right).\label{eq:f-theta-a}
\end{align}
If $Var(Y|X,A)=\sigma^{2}$ is constant, then it is fairly straightforward
to show that $Var\left(f_{\text{DR},\hat{\theta}}(Z)|X,\hat{\theta}=\theta\right)=\kappa(X)^{-1}\pi(X)^{-1}\sigma^{2}$
(Appendix \ref{sec:Variance-of-POs}). Again, extreme values of the
propensity score lead to regions where the pseudo-outcome has a high
variance. Inspired by this fact, we will see in the sections below
that using weights $\hat{\kappa}(X)\hat{\pi}(X)$ when fitting a POR
to predict $f_{\text{DR},\hat{\theta}}(Z)$ leads to fast convergence
rates and better simulated errors.

Table \ref{tab:Different-available-pseudo-outco} summarizes the above
relationships.

\begin{table*}
\caption{\label{tab:Different-available-pseudo-outco}Different available pseudo-outcome
transformations and their conditional variances given $X$, under
certain simplifying assumptions (see Appendix \ref{sec:Variance-of-POs}).}

\vskip 0.15in%
\begin{tabular}{|c|c|c|}
\hline 
Label & Outcome Transformation & Conditional Variance\tabularnewline
\hline 
\hline 
U, R ($f_{\text{U},\theta})$ & $\frac{Y-\eta(X)}{A-\pi(X)}$ & $\propto\frac{\pi^{3}+\left\{ 1-\pi\right\} ^{3}}{\left(1-\pi\right)^{2}\pi^{2}}=\mathbb{E}\left[\left(A-\pi(X)\right)^{-2}|X\right]$\tabularnewline
\hline 
DR ($f_{\text{DR},\theta})$ & $\mu_{1}(X)-\mu_{0}(X)\hspace{1em}\hspace{1em}\hspace{1em}\hspace{1em}\hspace{1em}\hspace{1em}$ & $\propto1/\nu(X)$\tabularnewline
 & $\hspace{1em}\hspace{1em}+\frac{A-\pi(X)}{\pi(X)(1-\pi(X))}\left(Y-\mu_{A}(X)\right)$ & \tabularnewline
\hline 
Oracle-R $(f_{\text{OR},\theta})$ & $\frac{\left\{ A-\pi(X)\right\} \left\{ Y-\eta(X)\right\} }{\pi(X)(1-\pi(X))}$ & $\propto1/\nu(X)$\tabularnewline
\hline 
\end{tabular}
\centering{}
\end{table*}

\section{Simulations\label{sec:Simulations}}

The goal of this simulation section is to examine the role of weights
in POR. We include a total of 6 simulation scenarios, labeled A, B,
C, D, E \& F. The first four are experiments taken from \citet{Nie2020-ih},
with $|X|$ set equal to 10. Setting E is the ``low dimensional''
simulated example from \citet{Kennedy2022-da}. Setting F is the simple
illustrative example from Figure \ref{fig:illustration}. Table \ref{tab:sim-def}
presents each setting in detail, and Table \ref{tab:sim-desc} gives
a qualitative summary of each setting. The settings generally differ
in their complexity for the functions $\eta$, $\tau$ and $\pi$.

\begin{table*}
\caption{\label{tab:sim-def} Simulation Setting Details. Below we show the
covariate distribution, CATE function, and nuisance functions for
simulations A through F. The notation $\text{trim}_{a}(b)$ is shorthand
for $\min(\max(a,b),1-a)$, and the notation $(a)_{+}$ is shorthand
for $\max(a,0)$. Settings A-D use multivariate, $iid$ covariates
$X$ with a dimension of 10. Here, each element of $X$ follows the
distribution shown in the second column. Simulations E \& F use univariate
$X$. A qualitative description of these simulation settings is shown
in Table \ref{tab:sim-desc}.}

\centering{}\vskip 0.15in%
\begin{tabular}{|l|l|l|l|l|}
\hline 
Label & $X$ distr. & $\tau\left(x\right)$ & $\mathbb{E}\left[Y|X=x\right]$ & $\mathbb{E}\left[A|X=x\right]$\tabularnewline
\hline 
\hline 
A & $U(0,1)$ & $\frac{1}{2}x_{1}+\frac{1}{2}x_{2}$ & $\sin(\pi x_{1}x_{2})+2\left(x_{3}-\frac{1}{2}\right)^{2}$ & $\text{trim}_{0.1}\left\{ \sin(\pi x_{1}x_{2})\right\} $\tabularnewline
\hline 
B & $N(0,1)$ & $\begin{array}{c}
\log(1+e^{x_{2}})\\
+x_{1}
\end{array}$ & $\begin{array}{c}
\max\{0,x_{1}+x_{2},x_{3}\}\\
\hspace{1em}+\left(x_{4}+x_{5}\right)_{+}
\end{array}$ & $1/2$\tabularnewline
\hline 
C & $N(0,1)$ & 1 & $2\log\left(1+e^{x_{1}+x_{2}+x_{3}}\right)$ & $\frac{1}{1+e^{x_{2}+x_{3}}}$\tabularnewline
\hline 
D & $N(0,1)$ & $\begin{array}{l}
\left(\sum_{i=1}^{3}x_{i}\right)_{+}\\
-\left(x_{4}+x_{5}\right)_{+}
\end{array}$ & $\begin{array}{l}
\left(\sum_{i=1}^{3}x_{i}\right)_{+}\\
\hspace{1em}+\frac{1}{2}\left(x_{4}+x_{5}\right)_{+}
\end{array}$ & $\frac{1}{1+e^{-x_{1}}+e^{-x_{2}}}$\tabularnewline
\hline 
E & $U(-1,1)$ & 0 & $\begin{array}{l}
1(x_{1}\leq-.5)\frac{(x_{1}+2)^{2}}{2}\\
+1(x_{1}>.5)(x_{1}+0.125)\\
+\left(\frac{x_{1}}{2}+0.875\right)1\left(-\frac{1}{2}<x_{1}<0\right)\\
+\left\{ 1\left(0<x_{1}<\frac{1}{2}\right)\right.\\
\left.\hspace{1em}\times\left(-5\left(x_{1}-\frac{1}{5}\right)^{2}+1.075)\right)\right\} 
\end{array}$ & $0.1+(0.8x_{1})_{+}$\tabularnewline
\hline 
F & $U\left(\frac{1}{20},\frac{19}{20}\right)$ & 0 & 1 & $x_{1}$\tabularnewline
\hline 
\end{tabular}
\end{table*}

\begin{table*}
\centering{}\caption{\label{tab:sim-desc} Qualitative summary of the simulation settings
detailed in Table \ref{tab:sim-def}.}
\vskip 0.15in%
\begin{tabular}{|l|l|l|l|l|}
\hline 
Label & Description & $\tau\left(x\right)$ & $\mathbb{E}\left[Y|X=x\right]$ & $\mathbb{E}\left[A|X=x\right]$\tabularnewline
\hline 
\hline 
A & Simple effect & \texttt{\textcolor{green}{}}\texttt{\textcolor{green}{{\color{teal}Simple}}} & \texttt{\textcolor{green}{}}\texttt{\textcolor{green}{{\color{BrickRed}Complex}}} & \texttt{\textcolor{green}{}}\texttt{\textcolor{green}{{\color{BrickRed}Complex}}}\tabularnewline
\hline 
B & Randomized trial & \texttt{\textcolor{green}{}}\texttt{\textcolor{green}{{\color{BurntOrange}Moderate}}} & \texttt{\textcolor{green}{}}\texttt{\textcolor{green}{{\color{BurntOrange}Moderate}}} & \texttt{\textcolor{green}{}}\texttt{\textcolor{green}{{\color{ForestGreen}Constant}}}\tabularnewline
\hline 
C & Complex prognosis & \texttt{\textcolor{green}{}}\texttt{\textcolor{green}{{\color{ForestGreen}Constant}}} & \texttt{\textcolor{green}{}}\texttt{\textcolor{green}{{\color{BrickRed}Complex}}} & \texttt{\textcolor{green}{}}\texttt{\textcolor{green}{{\color{teal}Simple}}}\tabularnewline
\hline 
D & Unrelated arms & \texttt{\textcolor{green}{}}\texttt{\textcolor{green}{{\color{BurntOrange}Moderate}}} & \texttt{\textcolor{green}{}}\texttt{\textcolor{green}{{\color{BurntOrange}Moderate}}} & \texttt{\textcolor{green}{}}\texttt{\textcolor{green}{{\color{BurntOrange}Moderate}}}\tabularnewline
\hline 
E & Non-differentiable prognosis & \texttt{\textcolor{green}{}}\texttt{\textcolor{green}{{\color{ForestGreen}Constant}}} & \texttt{\textcolor{green}{}}\texttt{\textcolor{green}{{\color{BrickRed}Complex}}} & \texttt{\textcolor{green}{}}\texttt{\textcolor{green}{{\color{teal}Simple}}}\tabularnewline
\hline 
F & Simple illustration & \texttt{\textcolor{green}{}}\texttt{\textcolor{green}{{\color{ForestGreen}Constant}}} & \texttt{\textcolor{green}{}}\texttt{\textcolor{green}{{\color{ForestGreen}Constant}}} & \texttt{\textcolor{green}{}}\texttt{\textcolor{green}{{\color{teal}Simple}}}\tabularnewline
\hline 
\end{tabular}
\end{table*}

We implemented POR with two pseudo-outcome functions, $f_{U,\hat{\theta}}$
and $f_{\text{DR},\hat{\theta}}$. In each case we used 10-fold cross-fitting.
For example, for $f_{U,\hat{\theta}}$, we used 90\% of the data to
estimate the nuisance functions $\hat{\theta}$, evaluated and stored
$f_{U,\hat{\theta}}(Z_{i})$ for the remaining 10\%, and then repeated
this process 10 times with different fold assignments to obtain a
pseudo-outcome for every individual. We then fit a regression against
all of these pseudo-outcomes together. We used boosted trees to perform
all of our nuisance regressions, as well as the final regression predicting
pseudo-outcomes as a function of $X$.\footnote{Specifically, we used the lightgbm R package written by \citet{Shi2023-bh}.}

For each pseudo-outcome function, we considered a weighted and unweighted
version. For $f_{U,\hat{\theta}}$ we compare uniform weights (i.e.,
the U-Learner) against weights $\left\{ A-\hat{\pi}(X)\right\} ^{2}$
(i.e., the R-Learner). For $f_{\text{DR},\hat{\theta}}$ we compare
uniform against weights $\hat{\pi}(X)\hat{\kappa}(X)$ (see Table
\ref{tab:Different-available-pseudo-outco}). 

As a baseline comparator, we consider a ``T-Learner'' approach \citep{Kunzel2019-ni},
which entails separately fitting two estimates $\hat{\mu}_{1}$ and
$\hat{\mu}_{0}$ for $\mu_{1}$ and $\mu_{0}$ respectively and then
taking $\hat{\mu}_{1}(x_{\text{new}})-\hat{\mu}_{0}(x_{\text{new}})$
as an estimate of $\tau(x_{\text{new}})$. We used the same boosted
tree algorithm when fitting the T-Learner.

Figure \ref{fig:simulation-result} shows the results of 400 simulation
iterations. Weighted POR matched or outperformed unweighted POR in
every setting. Performance was similar across the two weighted POR
methods we considered. The T-Learner performed comparably to weighted
POR in Settings D, E \& F, but dramatically underperformed in Settings
A, B \& C.

\begin{figure}
\begin{centering}
\includegraphics[width=1\columnwidth]{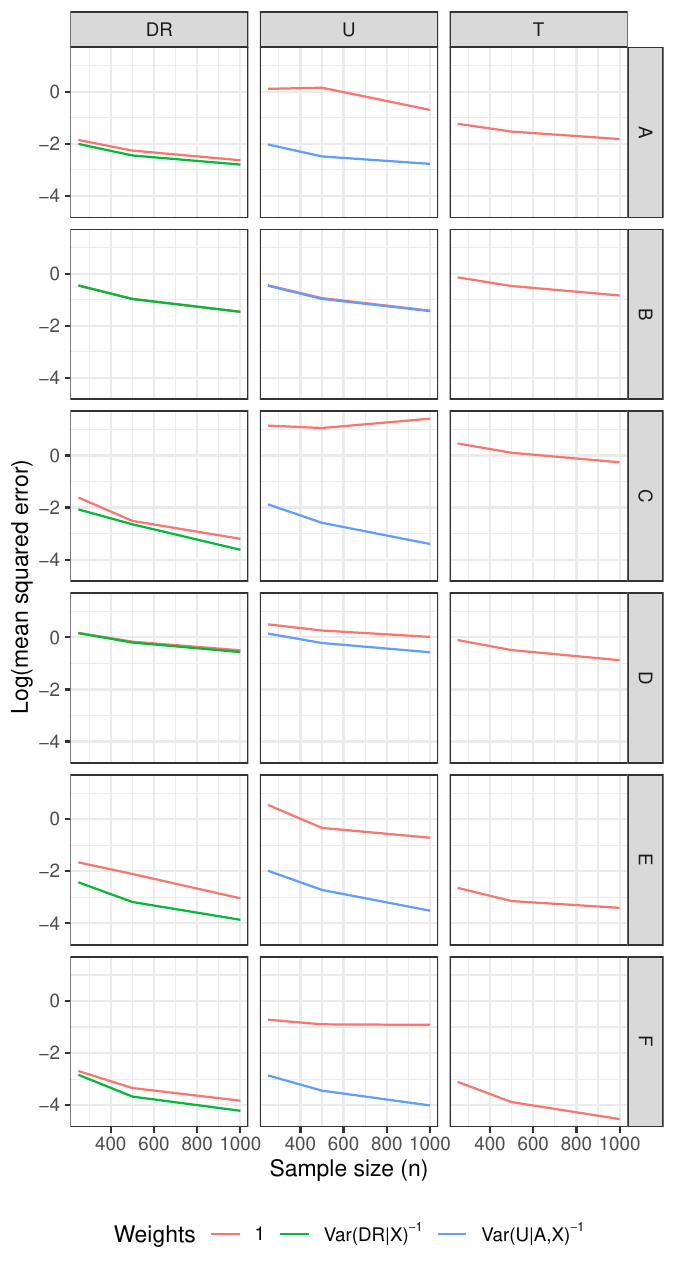}
\par\end{centering}
\caption{\label{fig:simulation-result} Weighted vs unweighted estimation of
simulated CATEs. The columns respectively represent POR with the DR-Learner
pseudo-outcome ($f_{\text{DR},\hat{\theta}}$), POR with the U-Learner
pseudo-outcome ($f_{\text{U},\hat{\theta}}$), and T-Learning. The
rows show the different simulation settings. For the weights, $\text{Var(U|A,X)}^{-1}$
is an abbreviation for $\left\{ A-\hat{\pi}(X)\right\} ^{2}\propto Var\left(f_{\text{U},\hat{\theta}}(Z)|A,X,\hat{\theta}\right)^{-1}$,
and $\text{Var(DR|X)}^{-1}$ is an abbreviation for $\hat{\pi}(X)\hat{\kappa}(X)\approx Var\left(f_{\text{DR},\hat{\theta}}(Z)|X,\hat{\theta}\right)^{-1}$. }
\end{figure}

\section{Convergence Rate Results\label{sec:theory}}

Part of the value the IVW framework is that it provides a straightforward
path for simplifying expressions for the bias of CATE estimates. Specifically,
if $Z,\hat{\kappa},\hat{\pi},$ and $\hat{\mu}$ are mutually independent,
we can make use of the following helpful identity. 
\begin{align}
 & \mathbb{E}\left(\hat{\kappa}\hat{\pi}\left(f_{1,\hat{\theta}}-f_{1,\theta}\right)|X\right)\nonumber \\
 & \hspace{1em}=\mathbb{E}\left(\hat{\kappa}\hat{\pi}A\left(\frac{1}{\hat{\pi}}-\frac{1}{\pi}\right)(\hat{\mu}_{1}-\mu_{1})|X\right)\nonumber \\
 & \hspace{1em}=\mathbb{E}\left(\hat{\kappa}\hat{\pi}\pi\left(\frac{1}{\hat{\pi}}-\frac{1}{\pi}\right)(\hat{\mu}_{1}-\mu_{1})|X\right)\nonumber \\
 & \hspace{1em}=\mathbb{E}\left(\hat{\kappa}|X\right)\mathbb{E}\left(\pi-\hat{\pi}|X\right)\mathbb{E}\left(\hat{\mu}_{1}-\mu_{1}|X\right).\label{eq:bias-identity}
\end{align}
The left-hand side is the weighted conditional bias in estimating
$f_{1,\hat{\theta}}$, which we can see depends only on the \emph{product}
of the biases for $\hat{\pi}$ and $\hat{\mu}$. The first equality
is shown in Appendix \ref{sec:Proof-of-Theorem}. The second equality
iterates expectations over $\{\hat{\mu}_{1},\hat{\kappa},\hat{\pi}\}$
to replace $A$ with $\pi$. The last comes from the independence
assumption. \citet{Kennedy2022-da} employs a similar identity when
reducing bias terms associated with the oracle R-Learner (see their
Section 7.6). In the remainder of this section, Eq (\ref{eq:bias-identity})
will play a fundamental role in our study of convergence rates.

\subsection{Notation\label{subsec:Notation}}

Let $\bar{\mathbf{Z}}=(\bar{\mathbf{X}},\bar{\mathbf{a}},\bar{\mathbf{y}})$
denote a dataset of $n$ observations used for POR, which we assume
is independent of the data used for estimating the nuisance functions
$\hat{\theta}$. Let $d$ denote the dimension of the domain $\mathcal{X}$
of $X$, and let $x_{\text{new}}$ be a point for which we would like
to predict $\tau(x_{\text{new}})$. 

We will often use the ``bar'' notation when referring to estimators
derived from $\bar{\mathbf{Z}}$; ``hat'' notation when referring
to quantities that depend on nuisance training data; and both notations
when referring to estimators derived from both datasets. We do this
to help keep track of dependencies between estimated quantities. Let
$\mathbf{X}_{\text{all}}$ be the combined matrix of covariates including
$\bar{\mathbf{X}}$ as well as the covariates used in training nuisance
functions.

Next we introduce notation to describe convergence rates. From random
variables $A_{n},B_{n}$, let $A_{n}\lesssim B_{n}$ denote that there
exists a constant $c$ such that $A_{n}\leq cB_{n}$ for all $n$.
Let $A_{n}\asymp B_{n}$ denote that $A_{n}\lesssim B_{n}$ and $B_{n}\lesssim A_{n}$.\textcolor{black}{{}
Let $A_{n}\lesssim_{\mathbb{P}}c_{n}$ denote that $A_{n}=O_{\mathbb{P}}(c_{n})$
for constants $c_{n}$.}\textcolor{teal}{}

We say that a function $f$ is $s$-smooth if there exists a constant
$c$ such that $|f(x)-f_{s,x'}(x)|\leq c||x-x'||^{s}$ for all $x,x'$,
where $f_{s,x'}$ is the $\lfloor s\rfloor^{th}$ order Taylor approximation
of $f$ at $x'$. This form of smoothness is a key property of functions
in a H\"{o}lder class (see, e.g., \citealp{Tsybakov2009-yb,Kennedy2022-da}).

For any function $g(Z)$, let $\bar{\mathbb{P}}_{n}(g(Z)):=\frac{1}{n}\sum_{i=1}^{n}g(Z_{i})$
denote its sample average over $\bar{\mathbf{Z}}$. We frequently
omit function arguments when clear from context, writing, for example,
$\bar{\mathbb{P}}_{n}(\pi)$ in place of $\bar{\mathbb{P}}_{n}(\pi(X))$.

\subsection{Setup \& Assumptions}

Following \citet{Kennedy2022-da}, we study convergence rates for
an estimator of $\tau$ that uses a local polynomial (LP) regression
for the POR step. To define this LP regression, let $h$ be a bandwidth
parameter that we expect will shrink with $n$, let $\texttt{kern}$
be a bounded, nonnegative kernel function that is zero outside of
the range {[}-1,1{]}, and let $K(X):=\frac{1}{h^{d}}\texttt{kern}\left(\frac{\|X-x_{\text{new}}\|}{h}\right)$.
Let $b$ be a $L$-dimensional, polynomial basis function that is
bounded on $\mathcal{X}$. Given independent estimates $\hat{\pi},$
$\hat{\kappa}$ and $\hat{\mu}$, let $\hat{\nu}(X):=\hat{\pi}(X)\hat{\kappa}(X)$,
and let $f_{\text{DR},\hat{\theta}}(Z)=f_{1,\hat{\theta}}(Z)-f_{0,\hat{\theta}}(Z)$
be an observed proxy for the transformation $f_{\text{DR},\theta}$,
where
\begin{align*}
 & f_{a,\hat{\theta}}(Z)\\
 & \hspace{1em}=\hat{\mu}_{a}(X)+\frac{1(A=a)}{a\hat{\pi}(X)+(1-a)\hat{\kappa}(X)}\left(Y-\hat{\mu}_{a}(X)\right).
\end{align*}
 Let
\[
\hat{\bar{\tau}}(x_{\text{new}}):=\frac{1}{n}\sum_{i=1}^{n}\hat{\bar{w}}(X_{i})f_{\text{DR},\hat{\theta}}(Z_{i})
\]
be an estimate of $\tau(x_{\text{new}})$, where
\[
\text{\ensuremath{\hat{\bar{w}}(x):=b(x_{\text{new}})^{\top}\hat{\bar{\mathbf{Q}}}^{-1}b(x)K(x)\hat{\nu}(x)}}
\]
and
\[
\hat{\bar{\mathbf{Q}}}:=\frac{1}{n}\sum_{i=1}^{n}b(X_{i})\hat{\nu}(X_{i})K(X_{i})b(X_{i})^{\top}.
\]
Thus, $\hat{\bar{\tau}}(x_{\text{new}})$ is a weighted LP regression
predicting $f_{\text{DR},\hat{\theta}}(Z)$ from $X$, with stabilizing
weights $\hat{\nu}(X)$. Hereafter, with some abuse of notation,
we also use the term ``weights'' to refer to $\hat{\bar{w}}(X)$.

We study $\hat{\bar{\tau}}(x_{\text{new}})$ by comparing it against
an oracle counterpart using the same estimated weights $\hat{\bar{w}}$,
but using the true function $f_{\text{DR},\theta}$. That is, we define
the oracle estimate
\[
\hat{\bar{\tau}}_{\text{oracle}}(x_{\text{new}}):=\frac{1}{n}\sum_{i=1}^{n}\hat{\bar{w}}(X_{i})f_{\text{DR},\theta}(Z_{i}).
\]
Given $\hat{\pi}$ and $\hat{\kappa}$, this oracle estimate is a
weighted LP regression predicting $f_{\text{DR},\theta}(Z)$ from
$X$, evaluated at the point $X=x_{\text{new}}$.

Next, we present several assumptions. We reuse the notation ``$c$''
to refer to generic constants; the same constant need not satisfy
all assumptions.
\begin{assumption}
\emph{\label{assu:regularity} (Regularity) $\mathbb{E}\left(Y^{2}|A,X\right)$
is bounded.}
\end{assumption}

\emph{}
\begin{assumption}
\emph{\label{assu:positivity} (Positivity) There exists a constant
$c\in(0,1)$ such that, for all covariate values $x$, all $a\in\{0,1\}$,
and all sample sizes $n$, we have $c\leq\hat{\kappa}(x),\kappa(x),\hat{\pi}(x),\pi(x)<1-c$.}
\end{assumption}

\begin{assumption}
\emph{\label{assu:bias-var-nuisance}(Nuisance Error) There exists
a complexity parameter $k$ (e.g., the number of parameters a model)
and constants $c$, $s_{\mu}$ and $s_{\pi}$, such that, with probability
approaching 1, the sequences $\mathsf{V}_{k,n}:=ck/n$, $\mathsf{B}_{\pi,k}:=ck^{-s_{\pi}/d}$
and $\mathsf{B}_{\mu,k}:=ck^{-s_{\mu}/d}$ satisfy
\begin{align*}
Var(\hat{\pi}(x)|\mathbf{X}_{\text{all}}) & \leq\mathsf{V}_{k,n},\\
Var(\hat{\kappa}(x)|\mathbf{X}_{\text{all}}) & \leq\mathsf{V}_{k,n},\\
Var(\hat{\mu}_{a}(x)|\mathbf{X}_{\text{all}}) & \leq\mathsf{V}_{k,n},
\end{align*}
and
\begin{align*}
\mathbb{E}(\hat{\pi}(x)-\pi(x)|\mathbf{X}_{\text{all}}) & \leq\mathsf{B}_{\pi,k},\\
\mathbb{E}(\hat{\kappa}(x)-\kappa(x)|\mathbf{X}_{\text{all}}) & \leq\mathsf{B}_{\pi,k},\\
\mathbb{E}(\hat{\mu}_{a}(x)-\mu_{a}(x)|\mathbf{X}_{\text{all}}) & \leq\mathsf{B}_{\mu,k}
\end{align*}
for all $x$ and $a$. Above, we assume that $k$ grows with $n$,
and that $k<n$.}
\end{assumption}

The bias conditions of Assumption \ref{assu:bias-var-nuisance} will
typically require $\mu_{a}$ and $\pi$ to be $s_{\mu}$-smooth and
$s_{\pi}$-smooth respectively. The variance conditions typically
will require the complexity of the nuisance models (i.e., $k$) to
grow at a limited rate. For example, for spline estimators, they generally
require the design matrices to have stable eigenvalues with high probability.
This can be ensured by requiring $k\log(k)/n$ to converge zero (see,
e.g., \citealp{Tropp2015-dx,Belloni2015-tn,Newey2018-da}). 

\begin{assumption}
\emph{\label{assu:limited-h} (Limited bandwidth) $n>1/h^{d}$.}
\end{assumption}

Assumption \ref{assu:limited-h} is fairly minimal, and is made for
simplicity of presentation. Roughly speaking, it says that $n$ needs
to be at least as large as the number of $h$-diameter subregions
required to fully partition the covariate space.

\begin{assumption}
\emph{\label{assu:bounded-lambda-min} (Eigenvalue Stability) There
exists a constant $c>0$ such that $\lambda_{\min}\left(\hat{\bar{\mathbf{Q}}}\right)>c$
with probability approaching 1.}
\end{assumption}

Assumption \ref{assu:bounded-lambda-min} ensures that the weights
$\hat{\bar{w}}$ are bounded in probability. \citet{Kennedy2022-da}
makes a similar assumption in their Theorem 3.

\begin{assumption}
\emph{\label{assu:uniform-density} ($X$ Distribution) The density
of $X$ is approximately uniform in the sense that, for any $h>0$
and $x\in\mathcal{X}$, we have $\text{Pr}\left[\|X-x\|\leq h\right]\lesssim h^{d}$.}
\end{assumption}

\begin{assumption}
\emph{\label{assu:local-independence} (Local Nuisance Estimators)
There exists a constant $c$ such that $Cov(\hat{\pi}(x),\hat{\pi}(x'))=0$,
$Cov(\hat{\kappa}(x),\hat{\kappa}(x'))=0$, and $Cov(\hat{\mu}_{a}(x),\hat{\mu}_{a}(x'))=0$
for all $x,x',a$ satisfying $\|x-x'\|>ck^{-1/d}$.}
\end{assumption}

Assumption \ref{assu:local-independence} says that the nuisance models'
predictions for sufficiently far away points $x,x'$ depend on entirely
different training data. This is true, for example, in $r$-order
spline regression models that divide each dimension into $p$ partitions,
producing a total of $p^{d}$ neighborhoods and $k=p^{d}d^{r}$ parameters.
If the neighborhoods are approximately evenly sized and $\mathcal{X}$
is the unit hypercube, the maximum distance within a neighborhood
is $\left(\sum_{i=1}^{d}1/p^{2}\right)^{1/2}=d^{1/2}/p=d^{1/2+r/d}k^{-1/d},$
where the last equality comes from rearranging $k=p^{d}d^{r}$. Thus,
predictions for points $x,x'$ that are at least $d^{1/2+r/d}k^{-1/d}$
apart will be independent, as they are created from different neighborhoods
of training data.

\subsection{Convergence rate results}

The assumptions in the previous section allow us to characterize the
difference between $\hat{\bar{\tau}}(x_{\text{new}})$ and the oracle
estimate.
\begin{thm}
\label{thm:main-bnd}(Error with respect to oracle) Under Assumptions
\ref{assu:regularity}-\ref{assu:local-independence}, we have the
following results.
\begin{enumerate}
\item \label{enu:(4-way-CF)-If}(4-way CF) If $\hat{\pi},\hat{\kappa},\hat{\mu}$,
and \textbf{$\bar{\mathbf{Z}}$} are mutually independent, then
\[
\hat{\bar{\tau}}(x_{\text{new}})-\hat{\bar{\tau}}_{\text{oracle}}(x_{\text{new}})\lesssim_{\mathbb{P}}\sqrt{\frac{1}{nh^{d}}}+\mathsf{B}_{\mu}\mathsf{B}_{\pi}.\hphantom{+\frac{k^{1/2-s_{\mu}/d}}{n}+\frac{k^{1/2-s_{\pi}/d}}{n}+\frac{k}{n}}
\]
\item \label{enu:(3-way-CF)-If}(3-way CF) If $\hat{\pi},\hat{\mu}$ and
\textbf{$\bar{\mathbf{Z}}$} are mutually independent; $\hat{\kappa}(x)=1-\hat{\pi}(x)$;
and \\
Var$\left[\sup_{x}\left\{ \hat{\pi}(x)-\pi(x)\right\} ^{2}|\mathbf{X}_{\text{all}}\right]\lesssim k_{n}/n$
with probability approaching 1, then
\[
\hat{\bar{\tau}}(x_{\text{new}})-\hat{\bar{\tau}}_{\text{oracle}}(x_{\text{new}})\lesssim_{\mathbb{P}}\sqrt{\frac{1}{nh^{d}}}+\mathsf{B}_{\mu}\left(\mathsf{B}_{\pi}+\mathsf{V}_{k,n}\right).\hphantom{+\frac{k^{1/2/2-s_{\pi}/d}}{n}+\frac{k}{n}}
\]
\item \label{enu:(2-way-CF)-If}(2-way CF) If $\{\hat{\pi},\hat{\mu}\}\perp\bar{\mathbf{Z}}$
and $\hat{\kappa}(x)=1-\hat{\pi}(x)$, then
\begin{align*}
 & \hat{\bar{\tau}}(x_{\text{new}})-\hat{\bar{\tau}}_{\text{oracle}}(x_{\text{new}})\\
 & \hspace{1em}\lesssim_{\mathbb{P}}\sqrt{\frac{1}{nh^{d}}}+\left(\mathsf{B}_{\mu}+\sqrt{\mathsf{V}_{k,n}}\right)\left(\mathsf{B}_{\pi}+\sqrt{\mathsf{V}_{k,n}}\right).
\end{align*}
\end{enumerate}
\end{thm}

The three bounds given by Theorem \ref{thm:main-bnd} become less
powerful as we relax the independence assumptions. As in \citet{Newey2018-da}
and \citet{Kennedy2022-da}, the independence conditions can be ensured
via higher-order cross-fitting, or ``nested'' cross-fitting, in
which separate folds are used to estimate each nuisance function.
Higher order cross-fitting is typically impractical in small or moderate
sample sizes, as it requires that a smaller fraction of data points
be used to train each nuisance function. That said, the effect of
dividing our sample into smaller partitions will be asymptotically
dwarfed by the effect of a faster convergence rate.

Point \ref{enu:(2-way-CF)-If} makes the weakest assumptions and produces
the least powerful bound. It is similar to the bound in Lemma 2 of
\citealp{Nie2020-ih}. That is, Point \ref{enu:(2-way-CF)-If} implies
that $\hat{\bar{\tau}}(x_{\text{new}})-\hat{\bar{\tau}}_{\text{oracle}}(x_{\text{new}})\lesssim_{\mathbb{P}}1/\sqrt{nh^{d}}$
if the conditional rMSE of $\hat{\pi}(x)$ and $\hat{\mu}_{a}(x)$
are $\lesssim n^{-1/4}$. The $\sqrt{1/nh^{d}}$ term common to all
three bounds is a standard variance term associated with LP regression
(see, e.g., Proposition 1.13 of \citealp{Tsybakov2009-yb}, or Theorem
3 of \citealp{Kennedy2022-da}). The variance condition in Point \ref{enu:(3-way-CF)-If}
is similar to Assumption \ref{assu:bias-var-nuisance}, and we expect
it to hold in similar situations.

To bound the error of the oracle itself, we additionally assume the
following.
\begin{assumption}
\emph{\label{assu:(Oracle-smoothness)}The target function $\tau$
is $s_{\tau}$-smooth, and the basis $b$ is of order at least $\lfloor s_{\tau}\rfloor$.}
\end{assumption}

From here, fairly standard results for local polynomial regression
(e.g., \citealp{Tsybakov2009-yb}; see also \citealp{Kennedy2022-da})
imply the following result.
\begin{thm}
(Oracle error) \label{thm:oracle-bnd}Under Assumptions \ref{assu:regularity}-\ref{assu:local-independence}
and Assumption \ref{assu:(Oracle-smoothness)},
\[
\hat{\bar{\tau}}_{\text{oracle}}(x_{\text{new}})-\tau(x_{\text{new}})\lesssim_{\mathbb{P}}\sqrt{\frac{1}{nh^{d}}}+h^{s_{\tau}}.
\]
\end{thm}

Combining the results of Theorems \ref{thm:main-bnd} \& \ref{thm:oracle-bnd},
we see that 
\begin{equation}
\hat{\bar{\tau}}(x_{\text{new}})-\tau(x_{\text{new}})\lesssim_{\mathbb{P}}\sqrt{\frac{1}{nh^{d}}}+h^{s_{\tau}}+\mathsf{B}_{\mu}\mathsf{B}_{\pi}\label{eq:final-4-bnd}
\end{equation}
when $\hat{\pi},\hat{\kappa},\hat{\mu}$ and \textbf{\textit{$\bar{\mathbf{Z}}$}}
are mutually independent and Assumptions \ref{assu:regularity}-\ref{assu:local-independence}
and \ref{assu:(Oracle-smoothness)} hold. 

The bound in Eq (\ref{eq:final-4-bnd}) is at least as low as the
bound established by \citet{Kennedy2022-da}, which adds an additional
$\mathsf{B}_{\pi}^{2}$ term. Our bound is not as low as the minimax
bound established by \citet{Kennedy2022-cn}, although the latter
depends on a slightly stronger assumption. Roughly speaking, \citeauthor{Kennedy2022-cn}
assume approximate knowledge of the covariate distribution, which
replaces our need for the covariance estimator $\hat{\bar{\mathbf{Q}}}$
and allows the authors to replace our $\mathsf{B}_{\mu}\mathsf{B}_{\pi}$
term with $\mathsf{B}_{\mu}\mathsf{B}_{\pi}h^{s_{\mu}+s_{\pi}}$ (\citeyear{Kennedy2022-cn};
see their Eq (16)).

\section{Discussion}

We have argued that R-Learning implicitly employs a POR with stabilizing
weights, and that these weight are key to its success. We also consider
doubly robust estimators that incorporate IVW more directly, and show
that they can attain a convergence rate that is, to our knowledge,
the fastest available under our minimal assumptions (Eq (\ref{eq:final-4-bnd})).

The use of weighted regression highlights two fundamental differences
in the difficulty of estimating the CATE versus the ATE. The CATE
is harder to estimate than the ATE in the sense that it is inherently
a more complex target, and so it incurs a higher oracle error. Indeed,
if the underlying CATE function is sufficiently non-smooth, then the
oracle error erodes any advantage of using doubly robust methods over
plug-in (``T-Learner'') methods. However, roughly speaking, when
estimating the CATE we have the extra advantage of being able to use
IVW without inducing confounding bias, and so the (higher) oracle
error rate becomes easier to attain. Both differences disappear in
the homogeneous effect setting when the CATE is constant, in which
case IVW is a natural approach for improving the ATE estimate (see,
e.g., \citealp{Hullsiek2002-jc,Yao2021-od}).

Our work also highlights an important caveat for R-Learning, which
is that it requires all confounders to be used as inputs in any resulting
decision support tool. For example, consider the process of applying
R-Learning to observational study in order to build a tool to identify
patients who will benefit most from a treatment. Doctors using this
tool must have access to all variables $(X)$ that were used for confounding
adjustment in the study. If the study involved extensive lab tests,
then this requirement may not be feasible. Alternatively, if the study
adjusted for race and income, in addition to insurance status, then
doctors may face ethical concerns if they allow information about
a patient's race or income to influence their recommended treatments.
While this problem can be partially mitigated by fitting an additional
regression to predict the R-Learning estimate from a subset of allowed
decision factors $V$, R-Learning may still underperform due to the
fact that it internally estimates a target that is more complex than
is necessary. Here, approaches that directly estimate the coarsened
function $\mathbb{E}\left(\tau(X)|V\right)$ may improve accuracy
due to the low oracle error associated with estimating lower-dimensional
functions (see, e.g., \citealp{Fisher2023-ks}).

\section*{Acknowledgements}

The author is grateful for many conversations with Virginia Fisher
that inspired this manuscript, and for her thoughtful comments on
early drafts. This work also would not be possible without several
helpful conversations with Edward Kennedy. Many of the proofs in this
manuscript are based on those shown by \citet{Kennedy2022-da}. 

\bibliographystyle{icml2024}
\bibliography{_surv}

\appendix
\onecolumn

\section{\label{sec:Proof-of-Theorem}Proof of Theorem \ref{thm:main-bnd}}

Throughout the appendix, we will sometimes use \tpur{colored text}
when writing long equations to flag parts of an equation that change
from one line to the next (e.g., Line (\ref{eq:colors})). We use
I.E. as an abbreviation for ``iterating expectations.''
\begin{proof}
Throughout the sections below we will use the fact if $1_{n}A_{n}\lesssim_{\mathbb{P}}b_{n}$
and $1_{n}$ is an indicator satisfying $\text{Pr}(1_{n}=1)\rightarrow1$
(at any rate), then $A_{n}\lesssim_{\mathbb{P}}b_{n}$ as well. In
particular, we define $\hat{\bar{1}}$ to be the event that the inequalities
in Assumptions \ref{assu:bias-var-nuisance} and \ref{assu:bounded-lambda-min}
hold. By these same assumptions, $\text{Pr}(\hat{\bar{1}}=1)\rightarrow1$.
When attempting to bound any given term $A_{n}$ in probability, it
will be sufficient to bound $\hat{\bar{1}}A_{n}$.

We can now present a proof outline. First, we decompose the error
with respect to the oracle as
\begin{align*}
\hat{\bar{\tau}}(x_{\text{new}})-\hat{\bar{\tau}}_{\text{oracle}}(x_{\text{new}}) & =\bar{\mathbb{P}}_{n}\left\{ \hat{\bar{w}}\left(\left(f_{1,\hat{\theta}}-f_{0,\hat{\theta}}\right)-\left(f_{1,\theta}-f_{0,\theta}\right)\right)\right\} \\
 & =\bar{\mathbb{P}}_{n}\left\{ \hat{\bar{w}}\left(f_{1,\hat{\theta}}-f_{1,\theta}\right)\right\} -\bar{\mathbb{P}}_{n}\left\{ \hat{\bar{w}}\left(f_{0,\hat{\theta}}-f_{0,\theta}\right)\right\} .
\end{align*}

Due to the symmetry of the problem, proving that either one of the
above terms is bounded will be sufficient. Without loss of generality
(WLOG), we focus on the first term. After multiplying by $\hat{\bar{1}}$,
which does not change the bound, we have
\begin{align}
\hat{\bar{1}}\bar{\mathbb{P}}_{n}\left\{ \hat{\bar{w}}\left(f_{1,\hat{\theta}}-f_{1,\theta}\right)\right\}  & =\hat{\bar{1}}\bar{\mathbb{P}}_{n}\left[\hat{\bar{w}}\left\{ \hat{\mu}_{1}-\mu_{1}+\frac{A}{\hat{\pi}}\left(Y-\hat{\mu}_{1}\right)-\frac{A}{\pi}\left(Y-\mu_{1}\right)\right\} \right]\nonumber \\
 & =\hat{\bar{1}}\bar{\mathbb{P}}_{n}\left[\hat{\bar{w}}\left\{ \hat{\mu}_{1}-\mu_{1}\mpur{-\frac{A}{\pi}\hat{\mu}_{1}}+\mblue{\frac{A}{\pi}\mu_{1}}\right.\right.\nonumber \\
 & \hspace{1em}\hspace{1em}+\frac{A}{\hat{\pi}}Y\mred{-\frac{A}{\hat{\pi}}\mu_{1}}-\frac{A}{\pi}Y+\frac{A}{\pi}\mu_{1}\nonumber \\
 & \hspace{1em}\hspace{1em}\left.\left.-\frac{A}{\hat{\pi}}\hat{\mu}_{1}\mred{+\frac{A}{\hat{\pi}}\mu_{1}}\mpur{+\frac{A}{\pi}\hat{\mu}_{1}}-\mblue{\frac{A}{\pi}\mu_{1}}\right\} \right]\label{eq:colors}\\
 & =\hat{\bar{1}}\bar{\mathbb{P}}_{n}\left[\hat{\bar{w}}\left(1-\frac{A}{\pi}\right)(\hat{\mu}_{1}-\mu_{1})\right]\label{eq:mean01}\\
 & \hspace{1em}\hspace{1em}+\hat{\bar{1}}\bar{\mathbb{P}}_{n}\left[\hat{\bar{w}}A\left(\frac{1}{\hat{\pi}}-\frac{1}{\pi}\right)(Y-\mu_{1})\right]\label{eq:mean02}\\
 & \hspace{1em}\hspace{1em}-\hat{\bar{1}}\bar{\mathbb{P}}_{n}\left[\hat{\bar{w}}A\left(\frac{1}{\hat{\pi}}-\frac{1}{\pi}\right)(\hat{\mu}_{1}-\mu_{1})\right].\label{eq:DR-term}
\end{align}

Section \ref{subsec:Helpful-Lemmas}, below, shows that the weights
$\hat{\bar{w}}$ satisfy $\mathbb{E}\left(\hat{\bar{1}}\hat{\bar{w}}(X_{i})^{2}\right)\lesssim1/h^{d}$
(as in \citet{Kennedy2022-da}'s Lemma 1). Under the condition that
$(\hat{\pi},\hat{\kappa},\hat{\mu}_{1})\perp\bar{\mathbf{Z}}$, Section
\ref{subsec:Mean-Zero-Terms} shows that Lines (\ref{eq:mean01})
\& (\ref{eq:mean02}) are weighted averages of terms that are $iid$
and mean zero, conditional $\hat{\pi},\hat{\kappa},\hat{\mu}_{1}$
and $\bar{\mathbf{X}}_{\text{all}}$. It will follow that Lines (\ref{eq:mean01})
\& (\ref{eq:mean02}) have expected conditional variance bounded by
$1/\left(nh^{d}\right)$. Thus, Lines (\ref{eq:mean01}) \& (\ref{eq:mean02})
are 
\begin{equation}
\lesssim_{\mathbb{P}}\frac{1}{\sqrt{nh^{d}}}\label{eq:first-2-mean-0}
\end{equation}
 by Markov's Inequality (see Section \ref{subsec:Mean-Zero-Terms}
for details). This fact holds for all forms of independence considered
in Theorem \ref{thm:main-bnd} (Points \ref{enu:(4-way-CF)-If}, \ref{enu:(3-way-CF)-If}
\& \ref{enu:(2-way-CF)-If}), as it depends only on $(\hat{\pi},\hat{\kappa},\hat{\mu}_{1})\perp\bar{\mathbf{Z}}$.
As an aside, these same steps can be used to show the first equality
in Eq (\ref{eq:bias-identity}).

Line (\ref{eq:DR-term}) does \emph{not }have mean zero given $\hat{\pi},\hat{\kappa},\hat{\mu}_{1}$
and $\bar{\mathbf{X}}_{\text{all}}$, and so constitutes the bias
relative to the oracle. These terms are more challenging to tackle
due to the correlations between the $\hat{\bar{\mathbf{Q}}}$ matrix
(contained within $\hat{\bar{w}}$) and the $1/\hat{\pi}$ nuisance
estimate. However, we can separate these quantities using the Cauchy
Schwartz inequality. Line (\ref{eq:DR-term}) becomes 
\begin{align}
 & \hat{\bar{1}}\bar{\mathbb{P}}_{n}\left\{ \hat{\bar{w}}A\left(\frac{1}{\hat{\pi}}-\frac{1}{\pi}\right)(\hat{\mu}_{1}-\mu_{1})\right\} \nonumber \\
 & \hspace{1em}\hspace{1em}=\hat{\bar{1}}b(x_{\text{new}})^{\top}\hat{\bar{\mathbf{Q}}}^{-1}\bar{\mathbb{P}}_{n}\left\{ b(X_{i})K(X_{i})\hat{\nu}(X_{i})A_{i}\left(\frac{1}{\hat{\pi}}-\frac{1}{\pi}\right)(\hat{\mu}_{1}-\mu_{1})\right\}  &  & \text{def of }\hat{\bar{w}}\nonumber \\
 & \hspace{1em}\hspace{1em}\leq\hat{\bar{1}}\left\Vert \hat{\bar{\mathbf{Q}}}^{-1}b(x_{\text{new}})\right\Vert \;\left\Vert \bar{\mathbb{P}}_{n}\left\{ bK\hat{\nu}A\left(\hat{\pi}^{-1}-\pi^{-1}\right)(\hat{\mu}_{1}-\mu_{1})\right\} \right\Vert  &  & \text{Cauchy Schwartz}\nonumber \\
 & \hspace{1em}\hspace{1em}\lesssim\hat{\bar{1}}\left\Vert \bar{\mathbb{P}}_{n}\left\{ bK\hat{\nu}A\left(\hat{\pi}^{-1}-\pi^{-1}\right)(\hat{\mu}_{1}-\mu_{1})\right\} \right\Vert  &  & \text{def of }\hat{\bar{1}}\text{ \& }b\nonumber \\
 & \hspace{1em}\hspace{1em}=\left[\sum_{l=1}^{L}\hat{\bar{1}}\bar{\mathbb{P}}_{n}\left\{ b_{\ell}K\hat{\nu}A\left(\hat{\pi}^{-1}-\pi^{-1}\right)(\hat{\mu}_{1}-\mu_{1})\right\} ^{2}\right]^{1/2}\nonumber \\
 & \hspace{1em}\hspace{1em}\leq\sum_{l=1}^{L}\left|\hat{\bar{1}}\bar{\mathbb{P}}_{n}\left\{ b_{\ell}K\hat{\kappa}\hat{\pi}A\left(\hat{\pi}^{-1}-\pi^{-1}\right)(\hat{\mu}_{1}-\mu_{1})\right\} \right|,\label{eq:L-sum}
\end{align}
where the last $\leq$ comes from the definition of $\hat{\nu},$
and from the fact that $\sum_{j=1}^{J}a_{j}^{2}\leq\left(\sum_{j=1}^{J}a_{j}\right)^{2}$
for any nonnegative sequences of values $\{a_{j},\dots,a_{J}\}$.

Appealing to Markov's Inequality, we tackle Line (\ref{eq:L-sum})
by bounding the second moment of each summand. For Point \ref{enu:(4-way-CF)-If},
we use the fact that $\mathbb{E}(V^{2})=Var(V)+\mathbb{E}(V)^{2}$
for any random variable $V$ to bound
\begin{align}
 & \mathbb{E}\left[\mathbb{E}\left\{ \hat{\bar{1}}\bar{\mathbb{P}}_{n}\left\{ b_{\ell}K\hat{\kappa}\hat{\pi}A\left(\hat{\pi}^{-1}-\pi^{-1}\right)(\hat{\mu}_{1}-\mu_{1})\right\} ^{2}|\mathbf{X}_{\text{all}},\hat{\kappa}\right\} \right]\nonumber \\
 & \;\;=\mathbb{E}\left[\mathbb{E}\left\{ \hat{\bar{1}}\bar{\mathbb{P}}_{n}\left\{ b_{\ell}K\hat{\kappa}\hat{\pi}A\left(\hat{\pi}^{-1}-\pi^{-1}\right)(\hat{\mu}_{1}-\mu_{1})\right\} |\mathbf{X}_{\text{all}},\hat{\kappa}\right\} ^{2}\right]\label{eq:cond-bias-DR}\\
 & \;\;\;\;\;\;+\mathbb{E}\left[Var\left\{ \hat{\bar{1}}\bar{\mathbb{P}}_{n}\left\{ b_{\ell}K\hat{\kappa}\hat{\pi}A\left(\hat{\pi}^{-1}-\pi^{-1}\right)(\hat{\mu}_{1}-\mu_{1})\right\} |\mathbf{X}_{\text{all}},\hat{\kappa}\right\} \right]\label{eq:cond-var-DR}
\end{align}
Section \ref{subsec:Bounding-bias-DR} shows that Line (\ref{eq:cond-bias-DR})
is 
\[
\lesssim k^{-2(s_{\mu}+s_{\pi})/d}
\]
when $\hat{\pi}\perp\hat{\kappa}$, using steps similar to those in
Eq (\ref{eq:bias-identity}).\textcolor{red}{}

Section \ref{subsec:Bounding-var-DR} shows that Line (\ref{eq:cond-var-DR})
is $\lesssim1/\left(nh^{d}\right)$. Thus, Eq (\ref{eq:L-sum}) is
\begin{align*}
 & \lesssim_{\mathbb{P}}\sqrt{\frac{1}{nh^{d}}}+k^{-(s_{\mu}+s_{\pi})/d}.
\end{align*}
This, combined with Line (\ref{eq:first-2-mean-0}), completes the
proof of Point \ref{enu:(4-way-CF)-If}. 

Section \ref{subsec:Proof-of-3way} shows that Line (\ref{eq:L-sum})
is 
\[
\lesssim_{\mathbb{P}}k^{-(s_{\mu}-s_{\pi})/d}+\frac{k^{1-s_{\mu}/d}}{n}+\sqrt{\frac{1}{nh^{d}}}
\]
under the conditions of Point \ref{enu:(3-way-CF)-If}, and Section
\ref{subsec:Proof-of-2way} shows that Line (\ref{eq:L-sum}) is 
\[
\lesssim_{\mathbb{P}}\frac{k}{n}+\frac{k^{1/2-s_{\mu}/d}}{\sqrt{n}}+\frac{k^{1/2-s_{\pi}/d}}{\sqrt{n}}+k^{-\left(s_{\mu}+s_{\pi}\right)/d}
\]
under the conditions of Point \ref{enu:(2-way-CF)-If}. This completes
the proof for Points \ref{enu:(3-way-CF)-If} \& \ref{enu:(2-way-CF)-If}.
\end{proof}

\subsection{Bound on weights\label{subsec:Helpful-Lemmas}}

Here we show results for the weights $\hat{\bar{w}}$. Our approach
closely follows classic approaches for LP regression (e.g., \citealp{Tsybakov2009-yb};
see also \citealp{Kennedy2022-da}). Let $\mathcal{I}(x)=1(\|x-x_{\text{new}}\|\leq h)$,
so that $K(x)=0$ and $\hat{\bar{w}}(x)=0$ whenever $\mathcal{I}(x)=0$
by the definitions of $K$ and $\hat{\bar{w}}$.
\begin{lem}
(Bounded weights) \label{lem:(Bounded-weights)} Under Assumptions
\ref{assu:limited-h}, \ref{assu:bounded-lambda-min} \& \ref{assu:uniform-density}:
\begin{enumerate}
\item \label{enu:K1-vals}$K(X)\lesssim\frac{1}{h^{d}}\mathcal{I}(X)$,
and $\mathbb{E}\left(|K(X)|\right)\lesssim\frac{1}{h^{d}}\mathbb{E}\left(\mathcal{I}(X)\right)\lesssim1$;
\item \label{enu:K2} $\mathbb{E}\left[\left\{ \frac{1}{n}\sum_{i=1}^{n}|K(X_{i})|\right\} ^{2}\right]\lesssim1$;
\item \label{enu:max-w}$\hat{\bar{1}}|\hat{\bar{w}}(x)|\lesssim\mathcal{I}(x)/h^{d}$
for any fixed $x$;
\item \label{enu:exp-w}$\mathbb{E}\left\{ \hat{\bar{1}}|\hat{\bar{w}}(X_{i})|\right\} \lesssim1$;
and
\item \label{enu:exp-w2}$\mathbb{E}\left\{ \hat{\bar{1}}\hat{\bar{w}}(X_{i})^{2}\right\} \lesssim1/h^{d}$.
\end{enumerate}
\end{lem}

\begin{proof}
Point \ref{enu:K1-vals} comes immediately from the definitions of
$K$ and $\mathcal{I}$, and from Assumption \ref{assu:uniform-density}. 

For Point \ref{enu:K2},
\begin{align*}
\mathbb{E}\left[\left\{ \frac{1}{n}\sum_{i=1}^{n}|K(X_{i})|\right\} ^{2}\right] & \lesssim\frac{1}{n^{2}h^{2d}}\mathbb{E}\left[\left\{ \sum_{i=1}^{n}\mathcal{I}(X_{i})\right\} ^{2}\right] &  & \text{Point \ref{enu:K1-vals}}\\
 & =\frac{1}{n^{2}h^{2d}}\left[\mathbb{E}\left\{ \sum_{i=1}^{n}\mathcal{I}(X_{i})\right\} +\mathbb{E}\left\{ \sum_{i=1}^{n}\mathcal{I}(X_{i})\sum_{j\neq i}^{n}\mathbb{E}\left(\mathcal{I}(X_{j})|X_{i}\right)\right\} \right]\\
 & \lesssim\frac{1}{n^{2}h^{2d}}\left[nh^{d}+n(n-1)h^{2d}\right] &  & \text{Assm \ref{assu:uniform-density}}\\
 & =\frac{1}{nh^{d}}+\frac{1}{n^{2}}\left[n(n-1)\right]\\
 & \leq1. &  & \text{Assm \ref{assu:limited-h}.}
\end{align*}

For Point \ref{enu:max-w},
\begin{align*}
\hat{\bar{1}}|\hat{\bar{w}}(x)| & \leq\hat{\bar{1}}\|b(x_{\text{new}})\|\;\|\hat{\bar{\mathbf{Q}}}^{-1}b(x)K(x)\hat{\pi}(x)\| &  & \text{Cauchy Schwartz}\\
 & \lesssim\hat{\bar{1}}\|\hat{\bar{\mathbf{Q}}}^{-1}b(x)K(x)\hat{\nu}(x)\| &  & \text{def of }b\\
 & \leq\frac{\hat{\bar{1}}}{\lambda_{\min}\left(\hat{\bar{\mathbf{Q}}}\right)}\|b(x)K(x)\hat{\nu}(x)\|\\
 & \lesssim\|b(x)K(x)\hat{\nu}(x)\| &  & \text{def of }\hat{\bar{1}}\\
 & \leq|K(x)| &  & \text{def of }b,\text{Assm \ref{assu:positivity}}\\
 & \lesssim\frac{1}{h^{d}}\mathcal{I}(x) &  & \text{Point \ref{enu:K1-vals}.}
\end{align*}
Point \ref{enu:exp-w} follows from Points \ref{enu:K1-vals} \& \ref{enu:max-w}.
Similarly, for Point \ref{enu:exp-w2},
\[
\mathbb{E}\left\{ \hat{\bar{1}}\hat{\bar{w}}(X_{i})^{2}\right\} \lesssim\frac{1}{h^{2d}}\mathbb{E}\left\{ \mathcal{I}(x)\right\} \lesssim\frac{1}{h^{d}},
\]
where the first $\lesssim$ is from Point \ref{enu:max-w} and the
second is from Assumption \ref{assu:uniform-density}.
\end{proof}

\subsection{\label{subsec:Mean-Zero-Terms}Showing Lines (\ref{eq:mean01}) \&
(\ref{eq:mean02}) are $\lesssim_{\mathbb{P}}\sqrt{1/(nh^{d})}$}

Line (\ref{eq:mean01}) has conditional expectation 
\begin{align*}
 & \hat{\bar{1}}\mathbb{E}\left[\bar{\mathbb{P}}_{n}\left(\hat{\bar{w}}\left(1-\frac{A}{\pi}\right)(\hat{\mu}_{1}-\mu_{1})\right)|\bar{\mathbf{X}}_{\text{all}},\hat{\mu}_{1},\hat{\pi},\hat{\kappa}\right]\\
 & =\hat{\bar{1}}\bar{\mathbb{P}}_{n}\left(\hat{\bar{w}}\left(1-\frac{\mpur{\pi}}{\pi}\right)(\hat{\mu}_{1}-\mu_{1})\right)\\
 & =0
\end{align*}
and conditional variance
\begin{align*}
 & \hat{\bar{1}}Var\left[\bar{\mathbb{P}}_{n}\left(\hat{\bar{w}}\left(1-\frac{A}{\pi}\right)(\hat{\mu}_{1}-\mu_{1})\right)|\bar{\mathbf{X}}_{\text{all}},\hat{\mu}_{1},\hat{\pi},\hat{\kappa}\right]\\
 & \hspace{1em}=\frac{\hat{\bar{1}}}{n^{2}}\sum_{i=1}^{n}\hat{\bar{w}}(X_{i})^{2}(\hat{\mu}_{1}(X_{i})-\mu_{1}(X_{i}))^{2}\frac{1}{\pi(X_{i})^{2}}Var\left[A|\bar{\mathbf{X}}_{\text{all}}\right]\\
 & \hspace{1em}\lesssim\frac{\hat{\bar{1}}}{n^{2}}\sum_{i=1}^{n}\hat{\bar{w}}(X_{i})^{2}(\hat{\mu}_{1}(X_{i})-\mu_{1}(X_{i}))^{2} &  & \text{Assm \ref{assu:positivity}}\\
 & \hspace{1em}\lesssim_{\mathbb{P}}\frac{1}{n^{2}}\sum_{i=1}^{n}\mathbb{E}\left[\hat{\bar{1}}\hat{\bar{w}}(X_{i})^{2}\mathbb{E}\left\{ (\hat{\mu}_{1}(X_{i})-\mu_{1}(X_{i}))^{2}|\mathbf{X}_{\text{all}}\right\} \right] &  & \text{Markov's Ineq}\\
 & \hspace{1em}\lesssim\frac{1}{n^{2}}\sum_{i=1}^{n}\mathbb{E}\left[\hat{\bar{1}}\hat{\bar{w}}(X_{i})^{2}\right] &  & \text{def of }\hat{\bar{1}}\\
 & \hspace{1em}\lesssim\frac{1}{nh^{d}}. &  & \text{Lemma \ref{lem:(Bounded-weights)}.\ref{enu:exp-w2}.}
\end{align*}
Combining this with the fact that Line (\ref{eq:mean01}) is mean
zero given $\bar{\mathbf{X}}_{\text{all}},\hat{\mu}_{1},\hat{\pi}$,
and $\hat{\kappa}$ we have
\begin{align*}
 & \hat{\bar{1}}\mathbb{E}\left[\bar{\mathbb{P}}_{n}\left(\hat{\bar{w}}\left(1-\frac{A}{\pi}\right)(\hat{\mu}_{1}-\mu_{1})\right)^{2}|\bar{\mathbf{X}}_{\text{all}},\hat{\mu}_{1},\hat{\pi},\hat{\kappa}\right]\\
 & =\hat{\bar{1}}Var\left[\bar{\mathbb{P}}_{n}\left(\hat{\bar{w}}\left(1-\frac{A}{\pi}\right)(\hat{\mu}_{1}-\mu_{1})\right)^{2}|\bar{\mathbf{X}}_{\text{all}},\hat{\mu}_{1},\hat{\pi},\hat{\kappa}\right]\\
 & \lesssim_{\mathbb{P}}\frac{1}{nh^{d}},
\end{align*}
which implies that Line (\ref{eq:mean01}) is $\lesssim_{\mathbb{P}}\sqrt{\frac{1}{nh^{d}}}$
by Markov's Inequality (see Lemma 2 of \citealp{Kennedy2022-da} for
details).

Similarly, Line (\ref{eq:mean02}) has conditional expectation
\begin{align}
 & \mathbb{E}\left[\bar{\mathbb{P}}_{n}\left(\hat{\bar{w}}A\left(\frac{1}{\hat{\pi}}-\frac{1}{\pi}\right)(Y-\mu_{1})\right)|\bar{\mathbf{X}}_{\text{all}},\hat{\mu}_{1},\hat{\pi},\hat{\kappa}\right]\nonumber \\
 & =\bar{\mathbb{P}}_{n}\left[\hat{\bar{w}}\left(\frac{1}{\hat{\pi}}-\frac{1}{\pi}\right)\mathbb{E}\left\{ A(Y-\mu_{1})|X\right\} \right]\nonumber \\
 & =\bar{\mathbb{P}}_{n}\left[\hat{\bar{w}}\left(\frac{1}{\hat{\pi}}-\frac{1}{\pi}\right)\mathbb{E}\left\{ Y-\mu_{1}|X,A=1\right\} \pi(X)\right]\nonumber \\
 & =0.\label{eq:cond-exp-YA}
\end{align}

and conditional variance
\begin{align*}
 & Var\left[\bar{\mathbb{P}}_{n}\left(\hat{\bar{w}}A\left(\frac{1}{\hat{\pi}}-\frac{1}{\pi}\right)(Y-\mu_{1})\right)|\bar{\mathbf{X}}_{\text{all}},\hat{\mu}_{1},\hat{\pi},\hat{\kappa}\right]\\
 & \hspace{1em}=\frac{1}{n^{2}}\sum_{i=1}^{n}\hat{\bar{w}}(X_{i})^{2}\left(\frac{1}{\hat{\pi}(X_{i})}-\frac{1}{\pi(X_{i})}\right)^{2}Var\left[A(Y-\mu_{1})|\bar{\mathbf{X}}_{\text{all}}\right]\\
 & \hspace{1em}\lesssim\frac{1}{n^{2}}\sum_{i=1}^{n}\hat{\bar{w}}(X_{i})^{2} &  & \text{Assms \ref{assu:regularity} \& \ref{assu:positivity}}\\
 & \hspace{1em}\lesssim_{\mathbb{P}}\frac{1}{nh^{d}} &  & \text{Lemma \ref{lem:(Bounded-weights)}.\ref{enu:exp-w2} + Markov's Ineq.}
\end{align*}
Thus, the same reasoning implies that Line (\ref{eq:mean02}) is $\lesssim_{\mathbb{P}}\sqrt{\frac{1}{nh^{d}}}$.

\subsection{\label{subsec:Bounding-bias-DR}Showing Line (\ref{eq:cond-bias-DR})
is $\lesssim k^{-2\left(s_{\mu}+s_{\pi}\right)}$ when $\hat{\pi}\perp\hat{\kappa}$}

Let $\hat{1}$ be the indicator that the inequalities in Assumption
\ref{assu:bias-var-nuisance} hold, where $\hat{1}\geq\hat{\bar{1}}$,
and $\hat{1}$ depends only on $\mathbf{X}_{\text{all}}$. The inner
expectation in Line (\ref{eq:cond-bias-DR}) equals
\begin{align}
 & \mathbb{E}\left[\hat{\bar{1}}\bar{\mathbb{P}}_{n}\left\{ b_{\ell}K\hat{\kappa}\hat{\pi}A\left(\hat{\pi}^{-1}-\pi^{-1}\right)(\hat{\mu}_{1}-\mu_{1})\right\} |\mathbf{X}_{\text{all}},\hat{\kappa}\right]\nonumber \\
 & \leq\frac{\mblue{\hat{1}}}{n}\sum_{i=1}^{n}b_{\ell}(X_{i})K(X_{i})\hat{\kappa}(X_{i})\nonumber \\
 & \hspace{1em}\hspace{1em}\hspace{1em}\times\mathbb{E}\left[A\left(1-\frac{\hat{\pi}(X_{i})}{\pi(X_{i})}\right)|\mathbf{X}_{\text{all}}\right]\mathbb{E}\left[\hat{\mu}_{1}(X_{i})-\mu_{1}(X_{i})|\mathbf{X}_{\text{all}}\right] &  & \text{4-way independence}\label{eq:kap-pi}\\
 & \lesssim\frac{\hat{1}k^{-s_{\mu}}}{n}\sum_{i=1}^{n}|K(X_{i})|\;\left|\mathbb{E}\left[A\left(1-\frac{\hat{\pi}(X_{i})}{\pi(X_{i})}\right)|\mathbf{X}_{\text{all}}\right]\right| &  & \text{def of }\hat{1}\nonumber \\
 & =\frac{\hat{1}k^{-s_{\mu}}}{n}\sum_{i=1}^{n}|K(X_{i})|\;\left|\mathbb{E}\left[\mathbb{E}\left(A|\mathbf{X}_{\text{all}},\hat{\pi}\right)\left(1-\frac{\hat{\pi}(X_{i})}{\pi(X_{i})}\right)|\mathbf{X}_{\text{all}}\right]\right| &  & \text{I.E.}\nonumber \\
 & =\frac{\hat{1}k^{-s_{\mu}}}{n}\sum_{i=1}^{n}|K(X_{i})|\;\left|\mathbb{E}\left[\pi(X_{i})-\hat{\pi}(X_{i})|\mathbf{X}_{\text{all}}\right]\right| &  & \text{by }\mathbb{E}\left(A_{i}|\mathbf{X}_{\text{all}},\hat{\pi}\right)=\pi(X_{i})\nonumber \\
 & \lesssim\frac{k^{-s_{\mu}-s_{\pi}}}{n}\sum_{i=1}^{n}|K(X_{i})| &  & \text{def of }\hat{1}.\nonumber 
\end{align}
Note that Line (\ref{eq:kap-pi}) requires $\hat{\pi}(x)\perp\hat{\kappa}(x)$
in order to remove the conditioning on $\hat{\kappa}$ from the expectation
term containing $\hat{\pi}$. 

Thus, Line (\ref{eq:cond-bias-DR}) is 
\begin{align*}
 & \lesssim k^{-2\left(s_{\mu}+s_{\pi}\right)}\mathbb{E}\left[\left\{ \frac{1}{n}\sum_{i=1}^{n}|K(X_{i})|\right\} ^{2}\right]\lesssim k^{-2\left(s_{\mu}+s_{\pi}\right)}
\end{align*}
where the second $\lesssim$ comes from Lemma \ref{lem:(Bounded-weights)}.\ref{enu:K2}.

\subsection{\label{subsec:Bounding-var-DR}Showing Line (\ref{eq:cond-var-DR})
is $\lesssim1/(nh^{d})$}

Line (\ref{eq:cond-var-DR}) is the expected value of

\begin{align}
 & Var\left[\;\;\;\;\;\;\;\;\,\;\;\;\;\hat{\bar{1}}\bar{\mathbb{P}}_{n}\left\{ b_{\ell}K\hat{\kappa}A\left(1-\frac{\hat{\pi}}{\pi}\right)(\hat{\mu}_{1}-\mu_{1})\right\} \;|\;\mathbf{X}_{\text{all}},\hat{\kappa}\;\right]\nonumber \\
 & =Var\left[\;\;\mathbb{E}\left[\;\hat{\bar{1}}\bar{\mathbb{P}}_{n}\left\{ b_{\ell}K\hat{\kappa}A\left(1-\frac{\hat{\pi}}{\pi}\right)(\hat{\mu}_{1}-\mu_{1})\right\} \;|\;\mathbf{X}_{\text{all}},\hat{\pi},\hat{\kappa},\hat{\mu}_{1}\;\;\right]\;\;|\;\;\mathbf{X}_{\text{all}},\hat{\kappa}\right]\nonumber \\
 & \hspace{1em}+\mathbb{E}\left[Var\left[\hat{\bar{1}}\bar{\mathbb{P}}_{n}\left\{ b_{\ell}K\hat{\kappa}A\left(1-\frac{\hat{\pi}}{\pi}\right)(\hat{\mu}_{1}-\mu_{1})\right\} \;|\;\mathbf{X}_{\text{all}},\hat{\pi},\hat{\kappa},\hat{\mu}_{1}\;\;\right]\;\;|\;\;\mathbf{X}_{\text{all}},\hat{\kappa}\right] &  & \text{Law of Total Var}\nonumber \\
 & =Var\left[\;\frac{\hat{\bar{1}}}{n}\;\sum_{i=1}^{n}b_{\ell}K\hat{\kappa}\mblue{\left(\pi-\hat{\pi}\right)}(\hat{\mu}_{1}-\mu_{1})|\mathbf{X}_{\text{all}},\hat{\kappa}\right]\label{eq:apply-EK-sketch}\\
 & \hspace{1em}+\mathbb{E}\left[\frac{\hat{\bar{1}}}{n^{2}}\sum_{i=1}^{n}b_{\ell}^{2}K^{2}\hat{\kappa}^{2}\mblue{Var(A|\bar{\mathbf{X}})}\left(1-\frac{\hat{\pi}}{\pi}\right)^{2}(\hat{\mu}_{1}-\mu_{1})^{2}|\mathbf{X}_{\text{all}},\hat{\kappa}\right].\label{eq:small-var-term}
\end{align}

Section \ref{subsec:EK-sketch} shows that the expectation of Line
(\ref{eq:apply-EK-sketch}) is $\lesssim1/(nh^{d})$ and Section \ref{subsec:smaller-var-term}
shows that the expectation of Line (\ref{eq:small-var-term}) is $\lesssim1/(nh^{d})$.

\subsubsection{\label{subsec:EK-sketch}Showing the expectation of Line (\ref{eq:apply-EK-sketch})
is $\lesssim1/(nh^{d})$ \emph{}}

To study Line (\ref{eq:apply-EK-sketch}), it will be helpful to
introduce some abbreviations. Let $\epsilon_{\hat{\pi}i}:=\hat{\pi}(X_{i})-\pi(X_{i})$,
and $\epsilon_{\hat{\mu}i}:=\hat{\mu}_{1}(X_{i})-\mu_{1}(X_{i})$.
Line (\ref{eq:apply-EK-sketch}) becomes 
\begin{align}
 & Var\left[\frac{\hat{\bar{1}}}{n}\sum_{i=1}^{n}b_{\ell}(X_{i})K(X_{i})\hat{\kappa}(X_{i})\epsilon_{\hat{\pi}i}\epsilon_{\hat{\mu}i}|\mathbf{X}_{\text{all}},\hat{\kappa}_{i}\right]\nonumber \\
 & \lesssim\frac{\hat{1}}{n^{2}}\sum_{i=1}^{n}K(X_{i})^{2}Var\left(\epsilon_{\hat{\pi}i}\epsilon_{\hat{\mu}i}|\mathbf{X}_{\text{all}}\right)\label{eq:diag-var}\\
 & \hspace{1em}+\frac{\hat{1}}{n^{2}}\sum_{i=1}^{n}\sum_{j\in\{1,\dots n\}\backslash i}|K(X_{i})K(X_{j})|\;Cov\left(\epsilon_{\hat{\pi}i}\epsilon_{\hat{\mu}i},\epsilon_{\hat{\pi}j}\epsilon_{\hat{\mu}j}|\mathbf{X}_{\text{all}}\right),\label{eq:off-diag-cov}
\end{align}
by the definition of $b$.

To study these variance and covariance terms, we use the fact that
for any four variables $A_{1},A_{2},B_{1},B_{2}$ satisfying $(A_{1},A_{2})\perp(B_{1},B_{2})$,
we have
\begin{align}
 & Cov(A_{1}B_{1},A_{2}B_{2})\nonumber \\
 & =Cov(A_{1},A_{2})Cov(B_{1},B_{2})+\mathbb{E}(A_{1})\mathbb{E}(A_{2})Cov(B_{1},B_{2})+Cov(A_{1},A_{2})\mathbb{E}(B_{1})\mathbb{E}(B_{2}).\label{eq:cov-abab}
\end{align}
A corollary of Eq (\ref{eq:cov-abab}) is that
\begin{equation}
Var(A_{1}B_{1})=Var(A_{1})Var(B_{1})+\mathbb{E}(A_{1})^{2}Var(B_{1})+Var(A_{1})\mathbb{E}(B_{1})^{2}.\label{eq:Var-AB}
\end{equation}
Applying Eq (\ref{eq:Var-AB}), we see that Line (\ref{eq:diag-var})
equals

\begin{align}
 & \frac{\hat{1}}{n^{2}}\sum_{i=1}^{n}K(X_{i})^{2}\left\{ Var(\epsilon_{\hat{\pi}i}|\mathbf{X}_{\text{all}})Var(\epsilon_{\hat{\mu}i}|\mathbf{X}_{\text{all}})\right.\nonumber \\
 & \hspace{1em}\hspace{1em}\hspace{1em}\left.+\mathbb{E}(\epsilon_{\hat{\pi}i}|\mathbf{X}_{\text{all}})^{2}Var(\epsilon_{\hat{\mu}i}|\mathbf{X}_{\text{all}})+Var(\epsilon_{\hat{\pi}i}|\mathbf{X}_{\text{all}})\mathbb{E}(\epsilon_{\hat{\mu}i}|\mathbf{X}_{\text{all}})^{2}\right\} \nonumber \\
 & \lesssim\frac{1}{n^{2}}\sum_{i=1}^{n}K(X_{i})^{2} &  & \text{def of }\hat{1}.\label{eq:diag-terms}
\end{align}
For the off-diagonal terms in Line (\ref{eq:off-diag-cov}), we first
note that for any $i,j\in\{1,\dots n\}$ satisfying $i\neq j$ we
have
\begin{align}
 & \hat{1}Cov\left(\epsilon_{\hat{\pi}i}\epsilon_{\hat{\mu}i},\epsilon_{\hat{\pi}j}\epsilon_{\hat{\mu}j}|\mathbf{X}_{\text{all}}\right)\nonumber \\
 & =\hat{1}Cov\left(\epsilon_{\hat{\pi}i},\epsilon_{\hat{\pi}j}|\mathbf{X}_{\text{all}}\right)Cov\left(\epsilon_{\hat{\mu}i},\epsilon_{\hat{\mu}j}|\mathbf{X}_{\text{all}}\right)\nonumber \\
 & \hspace{1em}+\hat{1}Cov\left(\epsilon_{\hat{\pi}i},\epsilon_{\hat{\pi}j}|\mathbf{X}_{\text{all}}\right)\mathbb{E}\left(\epsilon_{\hat{\mu}i}|\mathbf{X}_{\text{all}}\right)^{2}+\hat{1}\mathbb{E}\left(\epsilon_{\hat{\pi}i}|\mathbf{X}_{\text{all}}\right)^{2}Cov\left(\epsilon_{\hat{\mu}i},\epsilon_{\hat{\mu}j}|\mathbf{X}_{\text{all}}\right) &  & \text{by Eq \eqref{eq:cov-abab}},\nonumber \\
 & \lesssim\hat{1}Cov\left(\epsilon_{\hat{\pi}i},\epsilon_{\hat{\pi}j}|\mathbf{X}_{\text{all}}\right)Cov\left(\epsilon_{\hat{\mu}i},\epsilon_{\hat{\mu}j}|\mathbf{X}_{\text{all}}\right)\nonumber \\
 & \hspace{1em}+\hat{1}Cov\left(\epsilon_{\hat{\pi}i},\epsilon_{\hat{\pi}j}|\mathbf{X}_{\text{all}}\right)+\hat{1}Cov\left(\epsilon_{\hat{\mu}i},\epsilon_{\hat{\mu}j}|\mathbf{X}_{\text{all}}\right) &  & \text{def of \ensuremath{\hat{1}}},\label{eq:cov-gam-pi}
\end{align}
where
\begin{align}
 & \hat{1}Cov(\epsilon_{\hat{\pi}i},\epsilon_{\hat{\pi}j}|\mathbf{X}_{\text{all}})\nonumber \\
 & \hspace{1em}=\hat{1}Cov(\epsilon_{\hat{\pi}i},\epsilon_{\hat{\pi}j}|\mathbf{X}_{\text{all}})1\left(\|X_{i}-X_{j}\|\leq ck^{-1/d}\right) &  & \text{Assm \ref{assu:local-independence}}\nonumber \\
 & \hspace{1em}\leq\hat{1}Var(\epsilon_{\hat{\pi}i}|\mathbf{X}_{\text{all}})^{1/2}Var(\epsilon_{\hat{\pi}j}|\mathbf{X}_{\text{all}})^{1/2}1\left(\|X_{i}-X_{j}\|\leq ck^{-1/d}\right) &  & \text{Cauchy Schwartz}\nonumber \\
 & \hspace{1em}\lesssim\frac{k}{n}1\left(\|X_{i}-X_{j}\|\leq ck^{-d}\right), &  & \text{def of \ensuremath{\hat{1}}}.\label{eq:cov-pi}
\end{align}
By the same reasoning, 
\begin{equation}
\hat{\bar{1}}Cov(\epsilon_{\hat{\mu}i},\epsilon_{\hat{\mu}j}|\mathbf{X}_{\text{all}})\lesssim\frac{k}{n}1\left(\|X_{i}-X_{j}\|\leq ck^{-1/d}\right).\label{eq:cov-gam}
\end{equation}
Plugging Eqs (\ref{eq:cov-pi}) \& (\ref{eq:cov-gam}) into Eq (\ref{eq:cov-gam-pi}),
we get
\begin{align}
\hat{1}Cov\left(\epsilon_{\hat{\pi}i}\epsilon_{\hat{\mu}i},\epsilon_{\hat{\pi}j}\epsilon_{\hat{\mu}j}|\mathbf{X}_{\text{all}}\right) & \lesssim\left(\frac{k^{2}}{n^{2}}+2\frac{k}{n}\right)1\left(\|X_{i}-X_{j}\|\leq ck^{-1/d}\right).\label{eq:Cov-4e}
\end{align}

Finally, plugging Eqs (\ref{eq:diag-terms}) \& (\ref{eq:Cov-4e})
into Lines (\ref{eq:diag-var}) \& (\ref{eq:off-diag-cov}), we see
that the expectation of the expectation of Line (\ref{eq:diag-var})
plus Line (\ref{eq:off-diag-cov}) is
\begin{align*}
 & \lesssim\mathbb{E}\left[\frac{1}{n^{2}}\sum_{i=1}^{n}K(X_{i})^{2}\right.\\
 & \hspace{1em}\hspace{1em}\left.+\frac{1}{n^{2}}\sum_{i=1}^{n}\sum_{j\in\{1,\dots n\}\backslash i}|K(X_{i})K(X_{j})|\;\frac{k}{n}1\left(\|X_{i}-X_{j}\|\leq ck^{-1/d}\right)\right]\\
 & \lesssim\frac{1}{n^{2}h^{2d}}\sum_{i=1}^{n}\mathbb{E}\left[\mathcal{I}(X_{i})\right]\\
 & \hspace{1em}\hspace{1em}+\frac{k}{n^{3}h^{2d}}\sum_{i=1}^{n}\sum_{j\in\{1,\dots n\}\backslash i}\mathbb{E}\left[\mathcal{I}(X_{i})\mathbb{E}\left\{ 1\left(\|X_{i}-X_{j}\|\leq ck^{-1/d}\right)|X_{i}\right\} \right] &  & \text{\text{Lemma \ref{lem:(Bounded-weights)}.\ref{enu:K1-vals}}}\\
 & \lesssim\frac{1}{n^{2}h^{2d}}\sum_{i=1}^{n}\mathbb{E}\left[\mathcal{I}(X_{i})\right]\\
 & \hspace{1em}\hspace{1em}+\frac{k}{n^{3}h^{2d}}\sum_{i=1}^{n}\sum_{j\in\{1,\dots n\}\backslash i}\mathbb{E}\left[\mathcal{I}(X_{i})k^{-1}\right] &  & \text{Assm \ref{assu:uniform-density}}\\
 & \lesssim\frac{1}{nh^{d}}+\frac{1}{nh^{d}}. &  & \text{Lemma \ref{lem:(Bounded-weights)}.\ref{enu:K1-vals}.}
\end{align*}
Thus, the expectation of Line (\ref{eq:apply-EK-sketch}) is $\lesssim1/(nh^{d})$
as well.

\subsubsection{\label{subsec:smaller-var-term}Showing the expectation of Line (\ref{eq:small-var-term})
is $\lesssim1/(nh^{d})$}

The expectation of Line (\ref{eq:small-var-term}) is
\begin{align*}
 & \lesssim\mathbb{E}\mathbb{E}\left[\frac{\hat{1}}{n^{2}}\sum_{i=1}^{n}K^{2}\hat{\kappa}^{2}\left(1-\frac{\hat{\pi}}{\pi}\right)^{2}(\hat{\mu}_{1}-\mu_{1})^{2}|\mathbf{X}_{\text{all}},\hat{\kappa}\right] &  & \text{def of }b\\
 & =\mathbb{E}\left[\frac{\hat{1}}{n^{2}}\sum_{i=1}^{n}K^{2}\hat{\kappa}^{2}\mathbb{E}\left\{ \left\{ \mblue{\frac{\pi}{\pi}}\left(1-\frac{\hat{\pi}}{\pi}\right)\right\} ^{2}|\mathbf{X}_{\text{all}}\right\} \mathbb{E}\left\{ (\hat{\mu}_{1}-\mu_{1})^{2}|\mathbf{X}_{\text{all}}\right\} \right] &  & \text{4-way independence}\\
 & =\mathbb{E}\left[\frac{\hat{1}}{n^{2}}\sum_{i=1}^{n}K^{2}\hat{\kappa}^{2}\mathbb{E}\left\{ \mblue{\frac{1}{\pi^{2}}\left(\pi-\hat{\pi}\right)^{2}}|\mathbf{X}_{\text{all}}\right\} \mathbb{E}\left\{ (\hat{\mu}_{1}-\mu_{1})^{2}|\mathbf{X}_{\text{all}}\right\} \right]\\
 & =\mathbb{E}\left[\frac{\hat{1}}{n^{2}}\sum_{i=1}^{n}K^{2}\mathbb{E}\left\{ \mblue{\left(\pi-\hat{\pi}\right)^{2}}|\mathbf{X}_{\text{all}}\right\} \mathbb{E}\left\{ (\hat{\mu}_{1}-\mu_{1})^{2}|\mathbf{X}_{\text{all}}\right\} \right] &  & \text{Assm \ref{assu:positivity}}\\
 & \lesssim\frac{1}{n^{2}}\sum_{i=1}^{n}\mathbb{E}\left[K(X_{i})^{2}\right] &  & \text{def of }\hat{1}\\
 & \lesssim\frac{1}{n^{2}h^{2d}}\sum_{i=1}^{n}\mathbb{E}\left[\mathcal{I}(X_{i})\right] &  & \text{Lemma \ref{lem:(Bounded-weights)}.\ref{enu:K1-vals}}\\
 & \lesssim\frac{1}{nh^{d}} &  & \text{Lemma \ref{lem:(Bounded-weights)}.\ref{enu:K1-vals}}.
\end{align*}

\subsection{\label{subsec:Proof-of-3way}Bounding Line (\ref{eq:L-sum}) under
the conditions of Point \ref{enu:(3-way-CF)-If}}

Here, we redefine $\hat{1}$ to additionally indicate that $Var(\hat{\pi}(x)^{2}|\mathbf{X}_{\text{all}})\leq ck/n$
for all $x$. By assumption, we still have $\text{Pr}(\hat{1}=1)\rightarrow1$.

We can add and subtract $\kappa(X)$ to see that the summands in Line
(\ref{eq:L-sum}) are 
\begin{align}
 & \leq\hat{1}|\bar{\mathbb{P}}_{n}\left\{ b_{\ell}K\left\{ \hat{\kappa}-\kappa\right\} \hat{\pi}A\left(\hat{\pi}^{-1}-\pi^{-1}\right)(\hat{\mu}_{1}-\mu_{1})\right\} |\label{eq:k-k+k}\\
 & \hspace{1em}\;\;\;\;\;\;+\hat{1}|\bar{\mathbb{P}}_{n}\left\{ b_{\ell}K\kappa\hat{\pi}A\left(\hat{\pi}^{-1}-\pi^{-1}\right)(\hat{\mu}_{1}-\mu_{1})\right\} |.\label{eq:k-true}
\end{align}
Since $\kappa(X)\perp\hat{\pi}(X)|X$, Line (\ref{eq:k-true}) can
be studied in the same way as in Sections \ref{subsec:Bounding-bias-DR}
\& \ref{subsec:Bounding-var-DR}, producing the same bound. We tackle
Line (\ref{eq:k-k+k}) by bounding its second moment, which is equal
to 
\begin{align}
 & \mathbb{E}\left[\;\;\;\;\;\;\;\;\;\;\;\mathbb{E}\left\{ \hat{1}\bar{\mathbb{P}}_{n}\left\{ b_{\ell}K\left\{ \hat{\kappa}-\kappa\right\} A\left(1-\frac{\hat{\pi}}{\pi}\right)(\hat{\mu}_{1}-\mu_{1})\right\} ^{2}|\mathbf{X}_{\text{all}}\right\} \;\;\right]\nonumber \\
 & =\mathbb{E}\left[\;\;\;\;\;\;\mathbb{E}\left\{ \hat{1}\bar{\mathbb{P}}_{n}\left\{ b_{\ell}K\left\{ \hat{\kappa}-\kappa\right\} A\left(1-\frac{\hat{\pi}}{\pi}\right)(\hat{\mu}_{1}-\mu_{1})\right\} |\mathbf{X}_{\text{all}}\right\} ^{2}\;\;\right]\label{eq:cond-bias-DR-3}\\
 & \;\;+\mathbb{E}\left[\;Var\left\{ \hat{1}\bar{\mathbb{P}}_{n}\left\{ b_{\ell}K\left\{ \hat{\kappa}-\kappa\right\} A\left(1-\frac{\hat{\pi}}{\pi}\right)(\hat{\mu}_{1}-\mu_{1})\right\} |\mathbf{X}_{\text{all}}\right\} \;\;\right].\label{eq:cond-var-DR-3}
\end{align}
For Line (\ref{eq:cond-bias-DR-3}), since $\hat{\kappa}(x)=1-\hat{\pi}(x),$
we have 
\[
\hat{\kappa}(x)-\kappa(x)=1-\hat{\pi}(x)-(1-\pi(x))=\pi(x)-\hat{\pi}(x),
\]
which implies that the inner expectation in Line (\ref{eq:cond-bias-DR-3})
equals
\begin{align*}
 & \frac{\hat{1}}{n}\sum_{i=1}^{n}b_{\ell}(X_{i})K(X_{i})\mathbb{E}\left\{ \hat{\mu}_{1}(X_{i})-\mu_{1}(X_{i})|X_{i}\right\} \\
 & \hspace{1em}\hspace{1em}\hspace{1em}\times\mathbb{E}\left\{ \left\{ \pi(X_{i})-\hat{\pi}(X_{i})\right\} A_{i}\left(1-\frac{\hat{\pi}(X_{i})}{\pi(X_{i})}\right)|\mathbf{X}_{\text{all}}\right\}  &  & \hat{\mu}\perp\hat{\pi}\\
 & =\frac{\hat{1}}{n}\sum_{i=1}^{n}b_{\ell}(X_{i})K(X_{i})\mathbb{E}\left\{ \hat{\mu}_{1}(X_{i})-\mu_{1}(X_{i})|X_{i}\right\} \\
 & \hspace{1em}\hspace{1em}\hspace{1em}\times\mathbb{E}\left\{ \left\{ \pi(X_{i})-\hat{\pi}(X_{i})\right\} ^{2}|\mathbf{X}_{\text{all}}\right\}  &  & \text{I.E. over \ensuremath{\hat{\pi}}}\\
 & \lesssim k^{-s_{\mu}/d}\left(k^{-2s_{\pi}/d}+\frac{k}{n}\right)\frac{1}{n}\sum_{i=1}^{n}|K(X_{i})| &  & \text{def of }\hat{1}\text{ \& }b_{\ell}.
\end{align*}
Thus, Line (\ref{eq:cond-bias-DR-3}) is
\begin{align}
 & \lesssim k^{-2s_{\mu}/d}\left(k^{-2s_{\pi}/d}+\frac{k}{n}\right)^{2}\mathbb{E}\left[\left\{ \frac{1}{n}\sum_{i=1}^{n}|K(X_{i})|\right\} ^{2}\right]\nonumber \\
 & \lesssim k^{-2s_{\mu}/d}\left(k^{-2s_{\pi}/d}+\frac{k}{n}\right)^{2} &  & \text{Lemma \ref{lem:(Bounded-weights)}.\ref{enu:K2}}\label{eq:DR3-bias2}
\end{align}
As in Section \ref{subsec:Bounding-var-DR}, Line (\ref{eq:cond-var-DR-3})
is the expected value of 
\begin{align}
 & Var\left[\;\;\;\;\;\;\;\;\,\;\;\;\;\hat{\bar{1}}\bar{\mathbb{P}}_{n}\left\{ b_{\ell}K\left(\pi-\hat{\pi}\right)A\left(1-\frac{\hat{\pi}}{\pi}\right)(\hat{\mu}_{1}-\mu_{1})\right\} \;|\;\mathbf{X}_{\text{all}}\;\right]\nonumber \\
 & =Var\left[\;\;\mathbb{E}\left[\;\hat{\bar{1}}\bar{\mathbb{P}}_{n}\left\{ b_{\ell}K\left(\pi-\hat{\pi}\right)A\left(1-\frac{\hat{\pi}}{\pi}\right)(\hat{\mu}_{1}-\mu_{1})\right\} \;|\;\mathbf{X}_{\text{all}},\hat{\pi},\hat{\mu}_{1}\;\;\right]\;\;|\;\;\mathbf{X}_{\text{all}}\right]\nonumber \\
 & \hspace{1em}+\mathbb{E}\left[Var\left[\hat{\bar{1}}\bar{\mathbb{P}}_{n}\left\{ b_{\ell}K\left(\pi-\hat{\pi}\right)A\left(1-\frac{\hat{\pi}}{\pi}\right)(\hat{\mu}_{1}-\mu_{1})\right\} \;|\;\mathbf{X}_{\text{all}},\hat{\pi},\hat{\mu}_{1}\;\;\right]\;\;|\;\;\mathbf{X}_{\text{all}}\right] &  & \text{Law of total var}\nonumber \\
 & \leq Var\left[\;\frac{\mblue{\hat{1}}}{n}\;\sum_{i=1}^{n}b_{\ell}K\mblue{\left(\pi-\hat{\pi}\right)^{2}}(\hat{\mu}_{1}-\mu_{1})|\mathbf{X}_{\text{all}}\right]\nonumber \\
 & \hspace{1em}+\mathbb{E}\left[\frac{\hat{1}}{n^{2}}\sum_{i=1}^{n}b_{\ell}^{2}K^{2}\left(\hat{\pi}-\pi\right)^{2}\mblue{Var(A|\bar{\mathbf{X}})}\left(1-\frac{\hat{\pi}}{\pi}\right)^{2}(\hat{\mu}_{1}-\mu_{1})^{2}|\mathbf{X}_{\text{all}}\right].\nonumber \\
 & =\hat{1}Var\left[\;\frac{1}{n}\;\sum_{i=1}^{n}b_{\ell}(X_{i})K(X_{i})\epsilon_{i\pi}^{2}\epsilon_{i\mu}|\mathbf{X}_{\text{all}}\right]\label{eq:apply-EK-sketch-3}\\
 & \hspace{1em}+\hat{1}\mathbb{E}\left[\frac{1}{n^{2}}\sum_{i=1}^{n}b_{\ell}(X_{i})K(X_{i})\epsilon_{i\pi}^{2}Var(A|\bar{\mathbf{X}})\left(1-\frac{\hat{\pi}(X_{i})}{\pi(X_{i})}\right)^{2}\epsilon_{i\mu}^{2}|\mathbf{X}_{\text{all}}\right],\label{eq:small-var-term-3}
\end{align}
where the last equality is by definition of $\epsilon_{i\pi}$ and
$\epsilon_{i\mu}$. Since $\hat{1}Var(\epsilon_{\hat{\pi}i}^{2}|\mathbf{X}_{\text{all}})\leq ck/n$
and $\epsilon_{\hat{\pi}i}^{2}\leq1$, we can follow the same steps
as in Section \ref{subsec:EK-sketch} (with $\left(\epsilon_{\hat{\pi}i},\epsilon_{\hat{\pi}j}\right)$
replaced throughout by $\left(\epsilon_{\hat{\pi}i}^{2},\epsilon_{\hat{\pi}j}^{2}\right)$)
to see that Line (\ref{eq:apply-EK-sketch-3}) has expectation $\lesssim1/(nh^{d})$.
Similarly, since $\epsilon_{\hat{\pi}i}^{2}\leq1$, we can follow
the same steps as in Section \ref{subsec:smaller-var-term} to see
that Line (\ref{eq:small-var-term-3}) has expectation $\lesssim1/(nh^{d})$.
Thus, by Markov's Inequality and Eq (\ref{eq:DR3-bias2}), we see
that Line (\ref{eq:k-k+k}) is 
\begin{align*}
 & \lesssim_{\mathbb{P}}k^{-s_{\mu}/d}\left(k^{-2s_{\pi}/d}+\frac{k}{n}\right)+\sqrt{\frac{1}{nh^{d}}}\\
 & \leq k^{-(s_{\mu}-s_{\pi})/d}+\frac{k^{1-s_{\mu}/d}}{n}+\sqrt{\frac{1}{nh^{d}}}.
\end{align*}

\subsection{\label{subsec:Proof-of-2way}Bounding Line (\ref{eq:L-sum}) under
the conditions of Point \ref{enu:(2-way-CF)-If}}

If we assume only that $(\hat{\pi},\hat{\mu}_{1})\perp\text{\textbf{Z}}$,
then
\begin{align}
 & \mathbb{E}\left[\left.\vphantom{\frac{1}{1}}\hat{\bar{1}}|\bar{\mathbb{P}}_{n}\left\{ b_{\ell}K\hat{\kappa}\hat{\pi}A\left(\hat{\pi}^{-1}-\pi^{-1}\right)(\hat{\mu}_{1}-\mu_{1})\right\} |\;\;\right|\mathbf{X}_{\text{all}}\right]\nonumber \\
 & \lesssim\hat{1}\bar{\mathbb{P}}_{n}\left\{ |K|\;\mathbb{E}\left(\left.\vphantom{a_{3_{3}}^{3^{3}}}\left|1-\hat{\pi}/\pi\right|\;\left|\hat{\mu}_{1}-\mu_{1}\right|\hspace{1em}\right|\;\mathbf{X}_{\text{all}}\right)\right\}  &  & A,b_{\ell}(x),\hat{\kappa}(x)\lesssim1\nonumber \\
 & \lesssim\hat{1}\bar{\mathbb{P}}_{n}\left\{ |K|\;\mathbb{E}\left(\pi\left|1-\hat{\pi}/\pi\right|\;\left|\hat{\mu}_{1}-\mu_{1}\right|\hspace{1em}|\;\mathbf{X}_{\text{all}}\right)\right\}  &  & \text{from \ensuremath{1/\pi(x)\lesssim1}}\nonumber \\
 & \leq\hat{1}\bar{\mathbb{P}}_{n}\left\{ |K|\;\mathbb{E}\left(\left(\pi-\hat{\pi}\right)^{2}|\mathbf{X}_{\text{all}}\right)^{1/2}\mathbb{E}\left(\left(\hat{\mu}_{1}-\mu_{1}\right)^{2}|\mathbf{X}_{\text{all}}\right)^{1/2}\right\}  &  & \text{Cauchy Schwartz}\nonumber \\
 & \lesssim\left(\frac{k}{n}+k^{-2s_{\mu}/d}\right)^{1/2}\left(\frac{k}{n}+k^{-2s_{\mu}/d}\right)^{1/2}\frac{1}{n}\sum_{i=1}^{n}|K(X_{i})| &  & \text{(\ensuremath{\hat{\pi}},\ensuremath{\hat{\mu}_{1}})\ensuremath{\perp\text{\textbf{Z}}}, and def. of \ensuremath{\hat{1}}}\nonumber \\
 & \lesssim\left(\sqrt{\frac{k}{n}}+k^{-s_{\mu}/d}\right)\left(\sqrt{\frac{k}{n}}+k^{-s_{\mu}/d}\right)\frac{1}{n}\sum_{i=1}^{n}|K(X_{i})|\label{eq:sqrt-ab}\\
 & \lesssim_{\mathbb{P}}\frac{k}{n}+\frac{k^{1/2-s_{\mu}/d}}{\sqrt{n}}+\frac{k^{1/2-s_{\pi}/d}}{\sqrt{n}}+k^{-\left(s_{\mu}+s_{\pi}\right)/d} &  & \text{Lemma \ref{lem:(Bounded-weights)}.\ref{enu:K1-vals} + Markov's Ineq.}\nonumber 
\end{align}
Above, Line \ref{eq:sqrt-ab} comes from the fact that $\sqrt{a+b}\leq\sqrt{a}+\sqrt{b}$
for any two positive constants $a,b$.

\section{\label{sec:oracle-proof}Proof of Theorem \ref{thm:oracle-bnd}}

First we remark that the ``reproducing'' property for local polynomial
estimators still holds even when $\hat{\nu}$ is pre-estimated. If
$f$ is a $\lfloor s_{\tau}\rfloor$ order polynomial, then there
exists a set of coefficients $\beta$ such that $f(x)=b(x)^{\top}\beta$.
Thus,
\begin{align}
f(x_{\text{new}})=b(x_{\text{new}})^{\top}\beta= & b(x_{\text{new}})^{\top}\hat{\bar{\mathbf{Q}}}^{-1}\sum_{i=1}^{n}b(X_{i})K(X_{i})\hat{\nu}(X_{i})b(X_{i})^{\top}\beta\nonumber \\
= & b(x_{\text{new}})^{\top}\hat{\bar{\mathbf{Q}}}^{-1}\sum_{i=1}^{n}b(X_{i})K(X_{i})\hat{\nu}(X_{i})f(X_{i})\nonumber \\
= & \sum_{i=1}^{n}\hat{\bar{w}}(X_{i})f(X_{i}).\label{eq:poly-repro}
\end{align}
Let $\tau(X_{i};x_{\text{new}})$ be the $\lfloor s_{\tau}\rfloor$
order Taylor approximation of $\tau$ at $x_{\text{new}}$. It follows
from Eq (\ref{eq:poly-repro}) that 
\begin{equation}
\frac{1}{n}\sum_{i=1}^{n}\hat{\bar{w}}(X_{i})\tau(X_{i};x_{\text{new}})=\tau(x_{\text{new}};x_{\text{new}})=\tau(x_{\text{new}}),\label{eq:expand-around}
\end{equation}
where the second equality comes from the fact that the Taylor approximation
is exact at $x_{\text{new}}$. 

Conditional on $\hat{\nu}$ and $\bar{\mathbf{X}}$, the oracle bias
is
\begin{align*}
 & \mathbb{E}\left(\left\{ \hat{\bar{\tau}}_{\text{oracle}}(x_{\text{new}})-\tau(x_{\text{new}})\right\} |\hat{\nu},\bar{\mathbf{X}}\right)\\
 & =\frac{1}{n}\sum_{i=1}^{n}\hat{\bar{w}}(X_{i})\mathbb{E}\left(f_{\text{DR},\theta}(Z_{i})|\hat{\nu},\bar{\mathbf{X}}\right)-\tau(x_{\text{new}})\\
 & =\frac{1}{n}\sum_{i=1}^{n}\hat{\bar{w}}(X_{i})\tau(X_{i})-\tau(x_{\text{new}}) &  & \text{\ensuremath{\hat{\nu}\perp f_{\text{DR},\theta}(Z_{i})|\bar{\mathbf{X}}}}\\
 & =\frac{1}{n}\sum_{i=1}^{n}\hat{\bar{w}}(X_{i})\left\{ \tau(X_{i})-\tau(X_{i};x_{\text{new}})\right\}  &  & \text{Eq \eqref{eq:expand-around}}\\
 & \leq\frac{1}{n}\sum_{i=1}^{n}|\hat{\bar{w}}(X_{i})|\;\left|\tau(X_{i})-\tau(X_{i};x_{\text{new}})\right|\;\left|\mathcal{I}(X_{i})\right| &  & \text{definitions of }\hat{\bar{w}}\text{ \& }\mathcal{I}\\
 & \leq\frac{1}{n}\sum_{i=1}^{n}|\hat{\bar{w}}(X_{i})|\;\left\Vert X_{i}-x_{\text{new}}\right\Vert ^{s_{\tau}}\;\left|\mathcal{I}(X_{i})\right| &  & \text{Assm \ref{assu:(Oracle-smoothness)}}\\
 & \leq\frac{h^{s_{\tau}}}{n}\sum_{i=1}^{n}|\hat{\bar{w}}(X_{i})| &  & \text{definition of }\mathcal{I}\\
 & \lesssim_{\mathbb{P}}h^{s_{\tau}} &  & \text{Lemma \ref{lem:(Bounded-weights)}.\ref{enu:exp-w} + Markov's Ineq.}
\end{align*}
The conditional variance of the oracle is 
\begin{align*}
Var\left(\hat{\bar{\tau}}_{\text{oracle}}(x_{\text{new}})|\hat{\nu},\bar{\mathbf{X}}\right) & =\frac{1}{n^{2}}\sum_{i=1}^{n}\hat{\bar{w}}(X_{i})^{2}Var(f_{\text{DR},\theta}(Z_{i})|X_{i})\\
 & \lesssim\frac{1}{n^{2}}\sum_{i=1}^{n}\hat{\bar{w}}(X_{i})^{2} &  & \text{Assms \ref{assu:regularity} \& \ref{assu:positivity} }\\
 & \lesssim_{\mathbb{P}}\frac{1}{nh^{d}} &  & \text{Lemma \ref{lem:(Bounded-weights)}.\ref{enu:exp-w2} + Markov's Ineq.}
\end{align*}
This, combined with a conditional version of Markov's Inequality (see
Lemma 2 of \citealp{Kennedy2022-da}), shows the result.

\section{\label{sec:Variance-of-POs}Conditional Variance of Pseudo-outcomes}

For the pseudo-outcome function $f_{\text{U,\ensuremath{\theta}}}$,
assume that $A\perp Y|X$ and $Var(Y|X)=\sigma^{2}$. It follows from
$A\perp Y|X$ that $\eta(X)=\mu_{1}(X)=\mu_{0}(X)$ and $Var(Y|X,A)=Var(Y|X)=\sigma^{2}$.
Thus, 
\begin{align*}
Var\left(f_{\text{U},\theta}(A,X,Y)|X\right) & =Var\left(\frac{Y-\eta(X)}{A-\pi(X)}|X\right)\\
 & =\mathbb{E}\left[Var\left(\frac{Y-\eta(X)}{A-\pi(X)}|X,A\right)|X\right]\\
 & \hspace{1em}\hspace{1em}\hspace{1em}+Var\left[\mathbb{E}\left(\frac{Y-\eta(X)}{A-\pi(X)}|X,A\right)|X\right] &  & \text{Law of Total Var}\\
 & =\mathbb{E}\left[\left(A-\pi(X)\right)^{-2}Var\left(Y|X,A\right)|X\right]\\
 & \hspace{1em}\hspace{1em}\hspace{1em}+Var\left[\frac{\mu_{A}(X)-\eta(X)}{A-\pi(X)}|X\right]\\
 & =\mathbb{E}\left[\left(A-\pi(X)\right)^{-2}|X\right]\sigma^{2}\\
 & \hspace{1em}\hspace{1em}\hspace{1em}+0 &  & \text{from }\eta(X)=\mu_{A}(X)\\
 & =\left\{ \frac{\pi(X)}{\left\{ 1-\pi(X)\right\} ^{2}}+\frac{1-\pi(X)}{\left\{ 0-\pi(X)\right\} ^{2}}\right\} \sigma^{2}\\
 & =\left\{ \frac{\pi^{3}+\left\{ 1-\pi\right\} ^{3}}{\left(1-\pi\right)^{2}\pi^{2}}\right\} \sigma^{2}.
\end{align*}

For $f_{\text{OR},\theta}(Z)$, if $A\perp Y|X$ and $\mathbb{E}\left[\left(Y-\eta(X)\right)^{2}|X\right]=\sigma^{2}$
then 
\begin{align*}
Var\left(f_{\text{OR},\theta}(Z)|X\right) & =\nu(X)^{-2}Var\left[\left(A-\pi(X)\right)\left(Y-\eta(X)\right)|X\right]\\
 & =\nu(X)^{-2}\mathbb{E}\left[\left(A-\pi(X)\right)^{2}\left(Y-\eta(X)\right)^{2}|X\right]\\
 & \hspace{1em}\hspace{1em}\hspace{1em}\hspace{1em}-\mathbb{E}\left[\left(A-\pi(X)\right)|X\right]^{2}\mathbb{E}\left[\left(Y-\eta(X)\right)|X\right]^{2}\\
 & =\nu(X)^{-2}\mathbb{E}\left[\left(A-\pi(X)\right)^{2}|X\right]\mathbb{E}\left[\left(Y-\eta(X)\right)^{2}|X\right]\\
 & =\nu(X)^{-1}\sigma^{2}.
\end{align*}

For $f_{\text{DR,\ensuremath{\theta}}}$, if $Var(Y|A,X)=\sigma^{2}$
we have
\begin{align*}
 & Var\left(f_{\text{DR,\ensuremath{\theta}}}(A,X,Y)|X\right)\\
 & =Var\left[\mu_{1}(X)-\mu_{0}(X)+\frac{A-\pi(X)}{\pi(X)(1-\pi(X))}\left(Y-\mu_{A}(X)\right)|X\right]\\
 & =\nu(X)^{-2}Var\left[\left(A-\pi(X)\right)\left(Y-\mu_{A}(X)\right)|X\right]\\
 & =\nu(X)^{-2}\left[Var\left\{ \vphantom{\frac{1}{1}}\left(A-\pi(X)\right)\mathbb{E}\left\{ \vphantom{\frac{1}{1}}Y-\mu_{A}(X)|A,X\right\} |X\right\} \right.\\
 & \hspace{1em}\hspace{1em}\hspace{1em}\hspace{1em}\left.\mathbb{E}\left\{ \vphantom{\frac{1}{1}}\left(A-\pi(X)\right)^{2}Var\left\{ \vphantom{\frac{1}{1}}Y-\mu_{A}(X)|A,X\right\} |X\right\} \right] &  & \text{Law of Total Var}\\
 & =\nu(X)^{-2}\left[0\right.\\
 & \hspace{1em}\hspace{1em}\hspace{1em}\hspace{1em}\left.\mathbb{E}\left\{ \vphantom{\frac{1}{1}}\left(A-\pi(X)\right)^{2}|X\right\} \sigma^{2}\right]\\
 & =\nu(X)^{-1}\sigma^{2}\\
 & =\kappa(X)^{-1}\pi(X)^{-1}\sigma^{2}.
\end{align*}
\textcolor{blue}{}\textcolor{purple}{}
\end{document}